\documentclass[12pt]{article}
\usepackage{algorithmic}
\usepackage{algorithm}
\usepackage{amsmath}
\usepackage{amssymb}
\usepackage{amsthm}
\usepackage{caption}
\usepackage{graphicx}
\usepackage{url} 

\usepackage{tikz}
\usetikzlibrary{shapes.geometric, arrows}
\tikzstyle{startstop} = [rectangle, rounded corners, minimum width=3cm, minimum height=1cm,text centered, draw=black, fill=red!30]
\tikzstyle{process} = [rectangle, minimum width=3cm, minimum height=1cm, text centered, draw=black, fill=orange!30]
\tikzstyle{decision} = [diamond, minimum width=3cm, minimum height=1cm, text centered, draw=black, fill=green!30]
\tikzstyle{arrow} = [thick,->,>=stealth]

\newtheorem{lemma}{Lemma}[section]
\newtheorem{proposition}{Proposition}[section] 
\numberwithin{equation}{section}

\newcommand{\by}{{\boldsymbol{y}}}
\newcommand{\bk}{{\boldsymbol{k}}}

\newcommand{\blind}{0}

\begin{document}

\def\spacingset#1{\renewcommand{\baselinestretch}%
{#1}\small\normalsize} \spacingset{1}


\if0\blind
{
  \title{\bf Multi-Frame Blind Manifold Deconvolution for Rotating Synthetic Aperture Imaging }
  \author{Daolin Li  and Jian Zhang\hspace{.2cm}\\
    School of Mathematics, Statistics and Actuarial Science, University of Kent\\
        Martin Benning \\
        Department of Computer Science, University College London\\}
  \maketitle
} \fi

\if1\blind
{
  \bigskip
  \bigskip
  \bigskip
  \begin{center}
    {\LARGE\bf Multi-Frame Blind Manifold Deconvolution for Rotating Synthetic Aperture Imaging}
\end{center}
  \medskip
} \fi

\bigskip
\begin{abstract}
  Rotating synthetic aperture (RSA) imaging system captures images of the target scene at different rotation angles by rotating a rectangular aperture. 
Deblurring acquired RSA images plays a critical role in reconstructing a latent sharp image underlying the scene. In the past decade, the emergence of blind convolution technology has revolutionised this field by its ability to model complex features from acquired images. Most of the existing methods attempt to solve the above ill-posed inverse problem through maximising a posterior.
 Despite this progress, researchers have paid limited attention to exploring low-dimensional manifold structures of the latent image within a high-dimensional ambient-space. Here,  we propose a novel method to process RSA images using manifold fitting and penalisation in the content of multi-frame blind convolution. We develop fast algorithms for implementing the proposed procedure.  Simulation studies demonstrate that manifold-based deconvolution can outperform conventional deconvolution algorithms in the sense that it can generate a sharper estimate of the latent image in terms of estimating pixel intensities and preserving structural details. 
\end{abstract}

\noindent%
\vfill

\newpage
\spacingset{1.75} 
\section{Introduction}
\label{sec:intro}

  Compact and lightweight satellites with their limited size restrict the use of large telescopes needed for high-resolution Earth imagery. 
Rotating synthetic aperture (RSA) imaging system, introduced by  Rafanelli and Rehfield (1993), captures images by rotating a rectangular aperture in time-varying angles. RSA could potentially enhance the resolution of all imagery, thereby maximising the performance of these constrained platforms.
 Due to the aperture's shape, images typically exhibit a  lower resolution along the short axis compared to the long axis with spatial asymmetry in the underlying point-spread function (PSF, or called blur kernel in the field of deconvolution), which changes over time as the mirror rotates  (Sun et al., 2023).  During rotation, the directional variability enables images taken at a different angle of the scene to capture unique features of the scene along the long axis (Zhi et al., 2021).
 Hence, it is crucial to employ an appropriate image-fusion method to synthesise information in the image sequence owing to the rotation of the mirror into a single high-resolution image.
 The synthesised images can match or even surpass the resolution of conventional circular apertures, offering a cost-effective and low-complexity alternative (Monreal et al., 2018). In contrast, circular apertures capture images with a uniform PSF, leading to consistent degradation across all angles. This uniformity restricts the capture of angle-dependent features in the image sequence. To achieve a high resolution with circular apertures,  larger lenses are required, incurring substantial manufacturing costs (Stepp et al., 2003). Recent studies in optical remote sensing have shown RSA's potential to reduce the payload size and weight constraints, which are often associated with traditional large-aperture systems (Sun et al., 2024; Lv at al., 2021). 
  To take a full advantage of RSA, advanced image fusion methods are required for reconstructing a high-resolution image from a plurality of image frames acquired by the strip aperture imaging sensor. Conventional fusion techniques, such as principal component analysis (Kwarteng and Chavez, 1989; Jelenek et al., 2016), wavelet transform (Li et al., 1995; Amolins et al., 2007), and wavelet decomposition (Cheng et al., 2015), often fall short due to limited adaptation to RSA's unique degradation characteristics. Current fusion software primarily relies on inputing multiple PSFs across various angles (Zackay et al., 2016), the resulting images may suffer from degradation caused by the rectangular aperture. Moreover, on-orbit satellite vibrations can further blur the acquired images and deteriorate RSA's performance.
 In literature, both non-blind convolution (Gregson et al., 2013; Dong et al., 2011) and blind deconvolution (Kotera et al., 2013; Sroubek and Milanfar, 2011) have been proposed for RSA image reconstruction.  Non-blind deconvolution techniques use a known PSF to deblur the acquired image set whereas
 blind deconvolution attempts to recover a sharp image from a set of blurred images captured by RSA when PSFs are unknown. As an ill-posed inverse problem, most blind deconvolution algorithms depend heavily on regularisation. A commonly used framework of regularisation is based on the maximum-a-posteriori (MAP) with pre-selected priors (Satish et al., 2020). 
The MAP allows for flexible integration of appropriate priors for the image and blur kernel at various stages, helping mitigate the ill-posed nature of the deconvolution problem and guiding the algorithm to reconstruct images  accurately (Levin et al., 2019).
  The MAP methods can be classified into two main types (Levin et al., 2019): (a) $\mbox{MAP}_{(x,k)}$ approach, estimating both the latent image and PSFs simultaneously; (b) marginal posterior approach, estimating the PSFs by its marginal posterior followed by using non-blind deconvolution algorithms to recover the latent image (Babacan et al., 2008; Wipf and Zhang, 2014). In both cases, the selection of priors plays an important role.   A wide range of priors have been considered by researchers, including hyper-Laplacian prior 
(Krishnan and Fergus, 2009), which leverages the sparsity of gradients in natural images (Miskin and Mackay, 2000); $L_0$-norm prior, which is effective for deblurring text images (Pan et al. 2016); total variation (\textbf{TV}) prior (Chan and Wong, 1998); $\frac{L_1}{L_2}$-norm prior (Krishnan et al., 2011); low-rank prior (Gu et al., 2014); and dark channel prior (He et al., 2010). Rameshan et al. (2012) pointed out that most blur kernels are not uniformly distributed and often exhibit distinct structures, such as radial or lateral symmetry. Commonly used priors for estimating blur kernels include $L_1$-norm prior (Dong et al., 2012), $L_2$-norm prior for a stabilised estimate (Ren et al., 2016),  $L_p$-norm in a  gradient-domain (Rameshan et al., 2012), and a specialised prior that encodes information on how blurred images relate to the sharp image (Liu et al., 2014). The RSA imaging system takes multiple images of the same target from various angels, which can reduce the ill-posed nature of the blind deconvolution problem. Studies on multi-frame blind deconvolution can be found in (Sroubek and Milanfar, 2011; Dong et al., 2012; Delbracio and Sapiro, 2015).

  Most existing deconvolution algorithms can be performed iteratively, whereby each iteration improves estimating PSFs and the latent image iteratively. 
 Iterative methods include MAP and expectation-maximization algorithms. Despite the improvement in image quality resulting from these advanced algorithms, blurs, generated by vibrations in the RSA, can still degrade the quality of RSA imaging due to unknown structures of the target scene. Here, to address the issue, we propose a manifold fitting method with the help of sparsity prior to reduce noise in captured images. Manifold fitting has been used in dimensionality reduction (Weinberger and Saul, 2006). It is based on the manifold hypothesis (Fefferman et al.,2016) that high-dimensional data often reside near a low-dimensional manifold within the surrounding ambient space. The blurred images obtained by an RSA system are embedded in a high-dimensional ambient space, often exhibiting a low-dimensional manifold structure (Osher et al., 2017; Gong et al., 2010). This motivates us to estimate a low-dimensional manifold that contains the latent image, into which we  project the acquired multi-frame images.  By leveraging manifold structures, we can significantly enhance the quality of acquired images.  Our procedure is implemented in three steps: (1) deconvoluting captured images and  promoting sparsity of the super-resolution images via an $L_1$-norm penalty; (2) using manifold fitting to improve the resolution of these images; (3) synthesising these enhanced  images to produce a high-resolution image for the scene.

  The main contributions of our proposed method are three-fold:
  \begin{itemize}
  	\item An muti-frame blind manifold convolution model is proposed.  An enhanced algorithm is developed to reduces the bias of manifold deconvolution on each captured image and therefore to improve the performance of subsequent manifold fitting.
           \item Fast algorithms are developed for solving the associated optimisation problems. 
\item  Simulation studies are conducted to demonstrate that combining blind deconvolution with manifold fitting can outperform both direct manifold modelling and direct blind convolution of the original captured images. Compared to the existing methods, the proposed procedure can produce a sharper, high-quality image from the acquired RSA images in terms of accurate pixel intensities and structural details.
  \end{itemize}

  The remainder of this paper is organised as follows: Section \ref{sec:prob} presents a detailed description and implementation of the proposed procedure.
  Section \ref{sec:simul} evaluates the effectiveness of the proposed method on a simulated dataset. Section \ref{sec:conc} concludes the paper with potential future improvements.
 Throughout the paper, the symbols $Y_i,K_i,X,N_i$ represent random variables, while the lowercase letters $y_i,k_i,x,n_i$ denote their associated specific values. By $\textbf{vec}(\cdot)$, we vectorise a multidimensional array by stacking its columns. Let $\ast$ denote a two-dimensional discrete convolution operator  applied to each channel individually as defined in Appendix \ref{Appendix A}.

\section{Methodology}
\label{sec:prob}
  The imaging process in the RSA system can be modelled as a multi-frame blind convolution of blur kernels (i.e., PSFs) with a high-resolution image  at different angles, plus some additive noise, namely,
  \begin{equation}
  	Y_i = K_i \ast X + N_i, \ i = 1,\cdots,n,
  	\label{eq2.1}
  \end{equation}
  where $K_i$,  $Y_i, X, N_i \in \mathbb{R}^{C \times a \times b}$ stand for unknown blur kernel, acquired blurred image, latent high-resolution image, and additive noise at angel $i$, respectively. The noise is usually assumed to be Gaussian. Here, $a$, $b$, $C$ and $n$ denote the height and width of the images,  the number of channels, and the number of rotating angles, respectively. For simplicity, we assume that blur kernels $\{K_i\}_{i=1}^n$ have the same, known dimensionality $s<< \min\{a, b\}$, and  that $K_i \in \mathbb{R}^{s \times s}$ lies on the probability simplex with a finite number of pixels as its support. $K_i$ represents a non-negative distribution of light and sums to one, maintaining the image's overall brightness. We assume that $X$ and $K_i$ are independent. 
Convoluting equation (\ref{eq2.1}) with $G_h$ and $G_v$, we have a  model for gradient images,
\begin{eqnarray}\label{eqmD}
\nabla_j Y_i=K_i\ast \nabla_j X+\nabla_j N_i, \ i=1,...,n; j\in \{h, v\},
\end{eqnarray}
 where $G_h = [1, -1]$ and $G_v = [1, -1]^{\mathbf{T}}$, and  for image $x$, $ \textbf{vec}(\nabla_h x)=G_h\ast x$ and $\textbf{vec}(\nabla_v x)=G_v\ast x$ denote the horizontal and vertical differences with respect to channel and pixel respectively. Determining the boundary of an object within a scene is referred to as edge detection.
For a natural image, edges are detected by the local variation in intensity which is usually determined by the intensity gradients based on differentiation operators. 

As pointed out before, multi-frame  blind deconvolution permits recovery of the latent image from a set of blurred images in the presence of unknown kernels, that is, estimating  
 $X$ using captured images $(Y_i)_{\le i\le n}$ when $(K_i)_{i=1}^n$ are not available. This is a severly ill-posed inverse problem as there are an infinite number of $(x, k)$ satisfying equations (\ref{eq2.1}) and (\ref{eqmD}). To tackle the ill nature of the problem,  strong priors are required for both the sharp image and blur kernels in order to regularise the solution space.
To facilitate the description of  a sparse prior on blur kernels, we introduce  $Q_i\in \mathbb{R}^{s \times s}$,  generalised blur kernels with specific values $q_i$'s. We reparametrise $K_i$ via projecting $Q_i$ onto the probability simplex, namely
 \begin{equation}
  	K_i = \mathop{\arg\min}\limits\{ ||K - Q_i||_{\mathrm{F}}^2: K \in \mathbb{R}^{s \times s},\  \ K \ \text{is on the probability simplex}\}.
  	\label{eq2.3}
  \end{equation}

\subsection{MAP blind deconvolution: XK-procedure}
An important observation on blind deconvolution made by Levin et al. (2009)  is the
failure of the conventional MAP approach. These authors emphasised the strength of estimating the kernel alone compared to estimating both the image and the kernel simultaneously. Recent studies have shown the advantage of using gradient images and marginalisation over the original MAP.  
Among them, while some prior specifications may differ,  the basic idea of the improved MAP is to estimate $x$ and $k_i$ by alternatively running the optimisation in the following two steps (Zhou et al., 2021), based on multi-frame  images $\by=(y_1,\cdots,y_n)$ captured by the RSA system.

{\bf $X$-Step}: Fixing the value of $\bk=(k_1,...,k_n)$,  maximise the penalised least squares function
\begin{eqnarray}\label{eqx}
l(x|\by,\bk)&=& \left\{\lambda_1\sum_{i=1}^n||y_i-k_i\ast x ||^2 - \log p_X(x)\right\},
\end{eqnarray}
  where $p_{X}(x)$ is a prior density functions of $X$.  

Note that gradients in sharp natural images generally exhibit sparsity (Miskin and MacKay, 2000), that is,  most gradients are near zero with only a few exceptions.  This can be described by a hyper-Laplacian distribution (Krishnan and MacKay, 2009). More general gradient features of image $x$ can be specified by
  \begin{equation}
  	\log p_X(x) \propto -\sum\limits_{j \in \{h,v\}} ||\textbf{vec}(\nabla_j x)||_{\alpha}^{\alpha}.
  	\label{eqxp}
  \end{equation}
  A sparse prior applies when $\alpha \leq 1$; a shrink prior applies when $\alpha=2$; and for natural images, typically $\alpha \in [0.5,0.8]$ (Simoncelli,1999). These terms in equation (\ref{eqxp}) thus take the form of $L_{\alpha}$-penalty terms $||\textbf{vec}(\nabla_h x)||_{\alpha}^{\alpha}$ and $||\textbf{vec}(\nabla_v x)||_{\alpha}^{\alpha}$  over the desired image. Combining equations (\ref{eqx}) and (\ref{eqxp}), the optimisation problem in equation (\ref{eqx}) can be reformulated as:
  \begin{equation}
  	\begin{split}
  	  \widehat{x} &= \mathop{\arg\min}\limits_{x} \bigg\{\lambda_1\sum\limits_{i=1}^{n} ||y_i - k_i \ast x||_{\mathrm{F}}^2 + \sum\limits_{j \in \{h,v\}} ||\textbf{vec}(\nabla_j x)||_{\alpha}^{\alpha} \bigg\}.\\
  	\end{split}
  	\label{eqxalp}
  \end{equation}

{\bf $K$-Step}:  Gradient images suppress low-frequency components and emphasise high-frequency details, such as edges, which are commonly present in natural images. Convolution has predominant effects on high-frequency components of $x$ with ignorable effects on low-frequency components of $x$. Therefore,  kernel estimation based on original images $(\by,x)$ may introduce some bias by using irrelevant low-frequency information. This suggests that estimating kernels in the gradient domain will be better than in the pixel domain (Cho and Lee, 2009). Accordingly, 
 fixing $x$, estimate $k_i$'s via  minimising the least square function for the gradient model (\ref{eqmD}),
  \begin{equation}
  	\begin{split}
  	  \widehat{q_i} &= \mathop{\arg\min}\limits_{q_i} \bigg(\sum\limits_{j \in \{h,v\}} ||\nabla_j y - q_i \ast \nabla_j \widehat{x}||_{\mathrm{F}}^2 + \mu ||\textbf{vec}(q_i)||_1 \bigg) \\
  	  \widehat{k_i} &=\min_{k_i}||q_i-k_i||_F^2\
          \mbox{  subject to that } k_i \mbox{ lies in the probability simplex.}
  	\end{split}
  	\label{eqQ}
  \end{equation}

  In the $X$-step, we estimate the latent image by imposing a penalty in the gradient image space.  However, the quality of the resulting image still depends on the noise level of the acquired images. Reducing noise in the acquired images $\{y_i\}_{i=1}^n$ can sharpen the reconstructed latent image $\widehat{x}$. 
To provide an insight into the noise reduction problem, we consider the optimisation problem
in equation (\ref{eqxalp}), without the penalty term $\sum\limits_{j \in \{h,v\}} ||\textbf{vec}(\nabla_j x)||_p^p$ temporarily.  Proposition \ref{proposition2.1} in the Appendix A shows that, in the frequency domain, the reconstructed sharp image can be expressed as a weighted average of the deconvolved images $\{\widetilde{x_i}\}_{i=1}^n$, where the deconvoluted image $\widetilde{x_i}$ is the deconvolution of the blurred image $y_i$ with its estimated blur kernel $\widehat{k_i}$ through the equation $\mathcal{F}(\widetilde{x_i}) \odot \mathcal{F}(\widetilde{k_i}) = \mathcal{F}(y_i)$. Therefore, reducing the noise in $\widetilde{x_i}$ can improve the quality of estimate $\widehat{x}$ of the latent image. 
This motivates us to enhance $\{\widetilde{x_i}\}_{i=1}^n$ for a better reconstruction of $\widehat{x}$.
  
Note that in the RSA image system, as the images are taken at a sequence of regular angels, the deconvolved images $\{\widetilde{x}_i\}_{i=1}^n$ may share some common features related to the latent image. This suggests an enhancement method that we replace each $\widetilde{x}_i$ with a weighted average (denoted by $\widehat{x}_i^{\ast}$) of nearby deconvolved images. This results in a refined estimate of the sharp image:
  \begin{equation}
  	\widehat{x} = \mathcal{F}^{-1} \bigg(\frac{\sum\limits_{i=1}^n \overline{\mathcal{F}(\widetilde{k}_i)} \odot \mathcal{F}(\widetilde{y}_i)}{\sum\limits_{i=1}^n \overline{\mathcal{F}(\widetilde{k}_i)} \odot \mathcal{F}(\widetilde{k}_i)} \bigg) = \mathcal{F}^{-1} \bigg(\sum\limits_{i=1}^n \frac{\overline{\mathcal{F}(\widetilde{k}_i)} \odot \mathcal{F}(\widetilde{k}_i)}{\sum\limits_{j=1}^n \overline{\mathcal{F}(\widetilde{k_j})} \odot \mathcal{F}(\widetilde{k_j})} \odot \mathcal{F}(\widehat{x_i^{\ast }}) \bigg), 
  	\label{eq2.12}
  \end{equation}
  where $\widetilde{y}_i$ is a denoised convolved image defined by the equation $\mathcal{F}(\widehat{x}_i^{\ast}) \odot \mathcal{F}(\widetilde{k}_i) = \mathcal{F}(\widetilde{y}_i)$
and $\widetilde{k}_i$ is the padded version of $\widehat{k}_i$ as detailed in the Appendix A. In the next subsection, we propose a new algorithm, where, for each $i$, we use $\widetilde{y_i}$ rather than taking a direct weighted average of the acquired images in a neighborhood of $y_i$. This is because a direct averaging may introduce bias in the denosing process due to kernel hetergeneity when the RSA  takes images at angels with varying PSFs. The above algorithm can further be improved by imposing a penalty term on gradient images when computing $\widetilde{x}_i$. 

\subsection{MAP blind manifold deconvolution: IMR-procedure}
\label{sec:meth}
Assume that the acquired images are noisy blind convolutions of PSFs with a sharp image which is embedded within a low-dimensional manifold.
 Following the idea described in the previous subsection, we develop an enhanced multi-frame blind deconvolution method by using manifold fitting, implemented in three main steps.
\begin{itemize}
 \item {\it IX-step}: Calculate deconvolved images $\{\widetilde{x}_i\}_{i=1}^n$ from the acquired RSA images $\{y_i\}_{i=1}^{n}$. That is, we map the acquired RSA images back to the latent source image space.
 \item {\it MF-step}: Calculate deconvolved manifold images using the results in the {\it IX}-step.. 
\item {\it RC-step}: Reconstruct latent high-resolution image. The convolution is reintroduced into  $\{\widetilde{x}_i\}_{i=1}^n$ and denoised convolved images $\{\widetilde{y}_i\}_{i=1}^n$ are generated, followed by a non-blind deconvolution. 
\end{itemize}

  \subsubsection{IX-step}
  To obtain the deconvolved images $\{\widetilde{x}_i\}_{i=1}^n$ from $\{y_i\}_{i=1}^{n}$, we apply the XK-procedure to calculate $\widehat{k}_i$,  its padded version
$\widetilde{k}_i$ and $\widehat{x}$.  
The alternative iterations will stop after running a pre-specified number of times, and the resulting estimates of blur kernels are recorded for calculating individual $\widehat{x}_i$ later. 

 Note that the objective function in equation (\ref{eqxalp}) is not convex.  Two auxiliary variables $\gamma_h \in \mathbb{R}^{C \times a \times b}$ and $\gamma_v \in \mathbb{R}^{C \times a \times b}$ are introduced under the framework of half-quadratic splitting (Zhou et al., 2021) (see Appendix \ref{Appendix B}). This method involves an Iteratively Reweighted Least Squares (IRLS) algorithm and a two-dimensional discrete Fourier transform, as detailed in Appendix \ref{Appendix A}. Given $\widehat{x}$, the optimisation problem  in equation (\ref{eqQ}) can be reformulated as a standard multivariate lasso problem in the Lagrangian form (see Appendix \ref{Appendix C}). The optimisation problem in equation (\ref{eqQ}) can then be efficiently solved using Celer, a state-of-the-art solver for the lasso regression  (Massias et al., 2018 and 2020). Projecting the estimated generalised kernels $\widehat{q}_i$ onto the probability simplex, we obtain the estimates $\widehat{k}_i$. Here, an $\mathcal{O}(s^2 \log s)$ algorithm of Wang and Carreira-Perpinan (2013) is used in the projection. 
 
 In the IX-step, we aim to restore the underlying image for each acquired image and to explore its low dimensional manifold structure. By the nature of the RSA imaging, each restored image only captures a part of the latent image.  For each $i$, we calculate the deconvolved image $\widetilde{x}_i$, based on the acquired image $y_i$ and the estimated blur kernel $\widetilde{k}_i$,  by penalisation, that is, solving the optimisation problem:  
  \begin{equation}
  	\widetilde{x}_i = \mathop{\arg\min}\limits_{x} \bigg(  \lambda_1||y_i - \widehat{k}_i \ast x||_{\mathrm{F}}^2  + \sum\limits_{j \in \{h,v\}} ||\textbf{vec}(\nabla_j x)||_{\alpha}^{\alpha} \bigg).
  	\label{eq3.1} 
  \end{equation}
This results in a better estimate than one derived from the equation  $\mathcal{F}(\widetilde{x_i}) \odot \mathcal{F}(\widetilde{k_i}) = \mathcal{F}(y_i)$.
  Like solving the equation (\ref{eqxalp}), we introduce auxiliary variables for the equation (\ref{eq3.1}) and perform the alternative minimisation as Zhou et al.(2021). The whole algorithm for computing $\{\widetilde{x}_i\}_{i=1}^n$ is outlined in Algorithm \ref{Algorithm3.1}. Compared to $(y_i)_{i=1}^n$, $(\widetilde{x}_i)_{i=1}^n$ are closer to the latent image. This is because $(y_i)_{i=1}^n$ are produced through the convolution of the latent image with heterogeneous kernels plus some additive noises. The influence of convolution on the observed image can be mitigated by solving the optimisation problem of equation (\ref{eq3.1}). 

  \begin{algorithm}[h]
  	\caption{XK-procedure for generate $\{\widetilde{x_i}\}_{i=1}^n$ from $\{y_i\}_{i=1}^n$}
  	\label{Algorithm3.1}
  	\renewcommand{\algorithmicrequire}{\textbf{Input:}}
  	\renewcommand{\algorithmicensure}{\textbf{Output:}}
  	
  	\begin{algorithmic}[1]
  		\REQUIRE blurred images $\{y_i\}_{i=1}^n$, kernel support $s$, regularisation weights $\lambda_1$ and $\mu$, number of iterations $T$, $\tau$
  		\ENSURE $\{\widetilde{x_i}\}_{i=1}^n$
  		
  		\STATE \textbf{Initialise} $\widehat{k_i}$ as a Gaussian kernel for $i = 1, 2, \dots, n$
  		
  		\FOR{iter $t = 1:T$}
	  		\STATE $\widehat{x} \leftarrow \widehat{x}^{t}$: compute $\widehat{x}^{t}$ by solving equation (\ref{eqxalp}) (using half-quadratic splitting).
	  		
	  		\FOR{i $= 1:n$}
		  		\STATE Reformulate the problem equation (\ref{eqQ}) in the form of equation (\ref{eqC.13}) and solve it to obtain $\widehat{q_i}$.
		  		
		  		\STATE $\widehat{k_i} \leftarrow \widehat{k_i}^{t}$ : compute $\widehat{k_i}^t$ as the projection of $\widehat{q_i}$ onto the probability simplex.
		  		
	  		\ENDFOR
	  		\IF {$\min\limits_i S_{\textbf{c}}[\textbf{vec}(\widehat{k_i}^{t}),\textbf{vec}(\widehat{k_i}^{t-1})] \geq \tau$ where $S_{\textbf{c}}$ denotes the Cosine similarity}
		  		\STATE \textbf{Break}
	  		\ENDIF
	  		\ENDFOR
  		\FOR{i $= 1:n$}
	  		\STATE Solve the problem equation (\ref{eq3.1}), using half-quadratic splitting
  		\ENDFOR
  	\end{algorithmic}
  \end{algorithm}

In the next step, manifold fitting is employed to denoise $\widetilde{x}_i$, bringing each deconvoved image much closer to $\mathcal{M}$, an estimated  manifold around the latent image.

  \subsubsection{MF-step}

The manifold fitting algorithm is based on the manifold hypothesis, that is, natural images, as high-dimensional data, often lie near a low-dimensional manifold within the ambient space, exhibiting an inherent low-dimensional structure (Fefferman et al., 2016;  Osher et al., 2017). Our objective is to estimate this low-dimensional manifold, which encapsulates the latent image within the high-dimensional space. The estimated manifold can then be used for denoising by projecting samples onto it (Gong et al., 2010). 
  
  Yao et al. (2023) developed an innovative method of manifold fitting that avoids iteration. The core idea is to express a refined sample as a weighted average of samples in a "rectangular" neighbourhood of the unrefined sample. This method requires an initial estimator of the latent manifold $\mathcal{M}$, which should not be too far from $\mathcal{M}$. Compared to the blurred image $y_i$, the deconvolved image $\widetilde{x_i}$ obtained in the previous step is much closer to the latent high-resolution image since the effect of convolution has been mitigated using the $\text{MAP}_{x,k}$ method. Therefore, $\{\widetilde{x_i}\}_{i=1}^n$ can serve as an initial estimator of the latent manifold. The noise in $\{\widetilde{x_i}\}_{i=1}^n$ can be further reduced through the manifold fitting method, which consists of two steps: 1) estimating the contraction direction; 2) local contraction.
  
  \textit{Estimation of contraction direction.} Expressing $\widetilde{x_i}$ as the vector form $z_i \in \mathbb{R}^D$ ($D = Cab$) and defining $z_i^{\ast} := \mathop{\arg\min}\limits_{z' \in \mathcal{M}} ||z' - z_i||_2$, the contraction direction is measured by $z_i^{\ast} - z_i$. A satisfying estimate of $z_i^{\ast} - z_i$ is $F(z_i) - z_i$ where
  \begin{equation}
  	F(z_i) = \sum\limits_{j=1, j \neq i}^n \alpha_j(z_i) z_j. 
  	\label{eq3.2}
  \end{equation}
  The weight $\alpha_j(z_i)$ is defined by
  \begin{equation}
  	\begin{gathered}
  		\widetilde{\alpha}_j(z_i)=
  		\begin{cases}
  			\big(1 - \frac{||z_i - z_j||_2^2}{r_1^2} \big)^k, \ ||z_i - z_j||_2 \leq r_1,\\
  			0, \ \text{otherwise},
  		\end{cases}
  		\alpha_j(z_i) = \frac{\widetilde{\alpha}_j(z_i)}{\sum\limits_{t=1, t\neq i}^n \widetilde{\alpha}_t(z_i)}.
  	\end{gathered}
  	\label{eq3.3}	
  \end{equation}
  Here, $k \geq 2$ is a constant integer, and the value of $r_1$ should be selected to ensure that there are not too many $z_j$'s near $z_i$.
  
  \textit{Local contraction.} Denote the contraction matrix of $z_i$ as $\widehat{U_i}$ and 
  \begin{equation}
  	\widehat{U_i} = \frac{[F(z_i) - z_i][F(z_i) - z_i]^{\textbf{T}}}{||F(z_i) - z_i||_2^2}. 
  	\label{eq3.4}
  \end{equation}
  Based on $\widehat{U_i}$, $z_j - z_i \ (j \neq i)$ can be decomposed into $u_{j,i}$ and $v_{j,i}$ with $z_j - z_i = u_{j,i} + v_{j,i}$, where
  \begin{equation}
  	u_{j,i} = \widehat{U_i} (z_j - z_i) \ \text{and} \ v_{j,i} = z_j - z_i - u_{j,i}. 
  	\label{eq3.5}
  \end{equation}
  Then, the estimate of $z^{\ast}_i$ is given by
  \begin{equation}
  	G(z_i) = \sum\limits_{j=1, j \neq i}^n \beta_j(z_i) z_j. 
  	\label{eq3.6}
  \end{equation}
  The weight $\beta_j(z_i)$ is defined by
  \begin{equation}
  	\begin{gathered}
  		w_u(u_{j,i}) = 
  		\begin{cases}
  			1, \ \text{if} \ ||u_{j,i}||_2 \leq \frac{r_2}{2}; \\
  			\bigg[1 - \big(\frac{2||u_{j,i}||_2 - r_2}{r_2}\big)^2 \bigg]^k, \ \text{if} \ ||u_{j,i}||_2 \in (\frac{r_2}{2},r_2); \\
  			0, \ \text{otherwise},\\
  		\end{cases}\\
  		w_v(v_{j,i})=
  		\begin{cases}
  			\bigg(1 - \frac{||v_{j,i}||_2^2}{r_1^2} \bigg)^k, \ \text{if} \ ||v_{j,i}||_2 \leq r_1; \\
  			0, \ \text{otherwise},\\
  		\end{cases}\\
  		\text{and} \ \beta_j(z_i) = \frac{w_u(u_{j,i})w_v(v_{j,i})}{\sum\limits_{t=1, t \neq i}^n w_u(u_{t,i})w_v(v_{t,i})}, \ \text{where} \ r_2 \gg r_1.
  	\end{gathered}
  	\label{eq3.7}
  \end{equation}
  Using this method, the estimated manifold is composed of $G(z_1), \cdots, G(z_n)$, where $G(z_i)$ actually represents the projection of $z_i$ onto the estimated manifold. In other words, $G(z_i)$ serves as the denoised version of $z_i$.  
  
  \subsubsection{RC-step}
  Our ultimate objective is to recover the latent high-resolution image. Although the samples projected onto the latent manifold, $\{G(z_i)\}_{i=1}^n$, are no longer affected by convolution and noise, they only represent partial features of the latent sharp image. We rewrite $G(z_i)$ in the form as $\widehat{x}^{\ast}_i \in \mathbb{R}^{c \times a \times b}$. The purpose of this step is to combine these projected samples and reconstruct the latent image. First, it is necessary to generate the refined images by reintroducing convolution into $\widehat{x}^{\ast}_i$ as follows:
  \begin{equation}
  	\widetilde{y}_i = \widehat{k}_i \ast \widehat{x}^{\ast}_i. 
  	\label{eq3.8}
  \end{equation} 
  
  Based on $\{\widehat{k}_i\}_{i=1}^n$ and $\{\widetilde{y}_i\}_{i=1}^n$, the sharp image is reconstructed by solving the optimisation problem
  \begin{equation}
  	\widehat{x} = \mathop{\arg\min}\limits_x \bigg( \lambda_1 \sum\limits_{i=1}^{n} ||\widetilde{y}_i - \widehat{k}_i \ast x||_F^2  + \sum\limits_{j \in \{h,v\}} ||\nabla_j \mathbf{x}||_{\alpha}^{\alpha} \bigg).
  	\label{eq3.9}
  \end{equation}
  Similarly, the half-quadratic splitting method is used to solve the problem (\ref{eq3.9}). The integrated algorithm is summarised in Algorithm \ref{Algorithm3.2}.
  \begin{algorithm}[!h]
  	\caption{IMR-procedure: Enhanced Multi-frame Blind Manifold Deconvolution}
  	\label{Algorithm3.2}
  	\renewcommand{\algorithmicrequire}{\textbf{Input:}}
  	\renewcommand{\algorithmicensure}{\textbf{Output:}}
  	
  	\begin{algorithmic}[1]
  		\REQUIRE blurred images $\{y_i\}_{i=1}^n$ and tuning parameters: $s$, $\lambda_1$, $\mu$, $T$, $\tau$, $r_1$ and $r_2$
  		\ENSURE estimated sharp image $\widehat{x}$
  		
  		\STATE \textbf{Initialise} $\widehat{k_i}$ as Gaussian kernel, for $i=1,2,\cdots,n$
  		
  		\STATE \textbf{Use} Algorithm \ref{Algorithm3.1} to obtain estimated blur kernels $\{\widehat{k_i}\}_{i=1}^n$ and deconvolved images $\{\widetilde{x_i}\}_{i=1}^n$ 
  		
  		\FOR{i $= 1:n$}
  		\STATE \textbf{Estimate Contraction Direction}: compute $F(z_i)$ as in equation (\ref{eq3.2}), with $\alpha_j(z_i)$ obtained from equation (\ref{eq3.3})
  		\STATE \textbf{Local Contraction}: utilise the contraction matrix from equation (\ref{eq3.4}) to decompose $z_j - z_i$, and then compute $\beta_j(z_i)$ and $G(z_i)$ using equations (\ref{eq3.7}) and  (\ref{eq3.6}) respectively.
  		\STATE \textbf{Reintroduce Convolution}: generate the enhanced convolved images $\{\widetilde{y}_i\}_{i=1}^n$ using equation (\ref{eq3.8}).  
  		\ENDFOR
  		
  		\STATE \textbf{Non-blind Deconvolution}: Solve the optimisation problem equation (\ref{eq3.9}) to obtain $\widehat{x}$.
  		
  	\end{algorithmic}
  \end{algorithm}

\section{Numerical results}
\label{sec:simul}
  To assess the effectiveness of the proposed method, we simulate blurred images from a real image and reconstruct the latent image using these blurred images.
  
  \subsection{Dataset}
  A colour (RGB) image (size: $512 \times 512$) of an astronaut is selected from \textit{skimage.data}, and its pixel intensities are rescaled from $\{0,1,\cdots,255\}$ to the range of $[0,1]$ (unless otherwise specified, all images are represented with pixel intensities in the range). Next, thirty-six sparse kernels, each of size $25 \times 25$, are generated to convolve with this RGB image, where each kernel represents a PSF at a unique orientation. Subsequently, \textit{i.i.d.} Gaussian noise with a standard deviation of 0.05 is added to these convolved images. The blurring step aligns to equation (\ref{eq2.1}), and an example of the resultant blurring is illustrated in Figure \ref{figure4.1}.
  \begin{figure}[!h]
  	\centering
  	\includegraphics[width=1.0\textwidth]{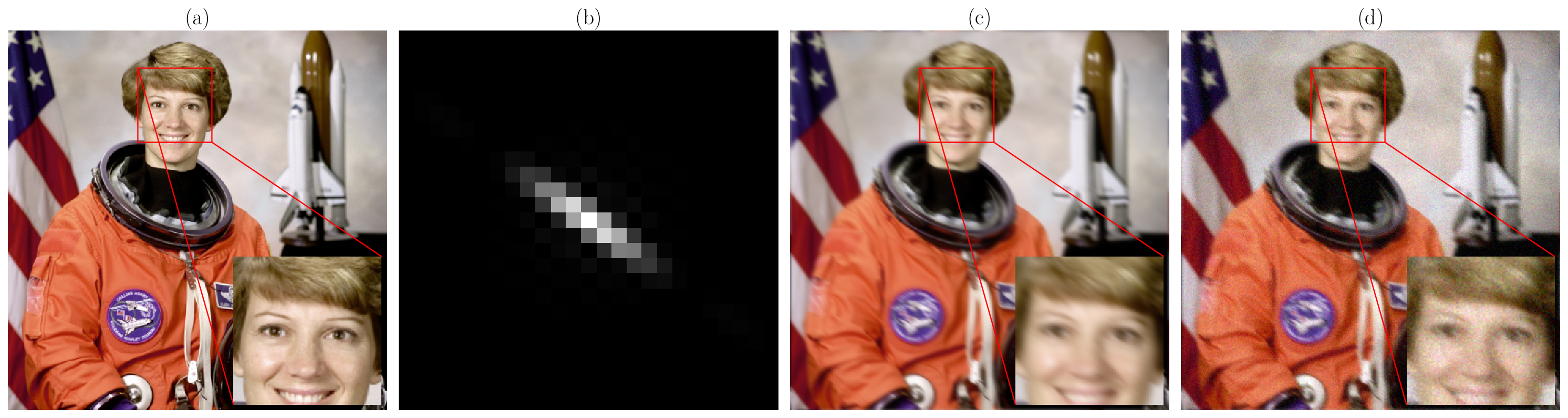}
  	\captionsetup{width=0.9\textwidth}
  	\caption{Process of blurring. (a) The ground-truth sharp image $x$; (b) one simulated PSF $k_i$; (c) one ground-truth convolved image $k_i \ast x$; (d) one blurred image $y_i = k_i \ast x + n_i$.}
  	\label{figure4.1}
  \end{figure}
  
  \subsection{Implementation details}
  When applying Algorithm \ref{Algorithm3.1}, we estimate the blur kernels first. The tuning parameters are set as follows: $\mu=0.04$, $\tau = 0.9980$, $\alpha = 0.80$, and $T=10$. And $\lambda_1=667$ for solving the optimisation problems in equations (\ref{eqxalp}) and (\ref{eqQ}) iteratively while $\lambda_1=2000$ for solving problem equation (\ref{eq3.1}).  The above values of $\lambda_1$ and $\mu$ have been pre-validated using a separate validation set. Figure \ref{figure4.2} shows the estimated kernels along with their corresponding ground-truth kernels. Using the blurred images $\{y_i\}_{i=1}^n$ and the estimated blur kernels $\{\widehat{k}_i\}_{i=1}^n$ (the second row of Figure \ref{figure4.2}), the deconvolved images $\{\widetilde{x}_i\}_{i=1}^n$ are obtained via non-blind deconvolution (solving equation (\ref{eq3.8})). Due to noise in $y_i$, many regions of $\widetilde{x}_i$ exhibit a lack of smoothness (see the first row of Figure \ref{figure4.3}).
  \begin{figure}[!h]
  	\centering
  	\includegraphics[width=1.0\textwidth]{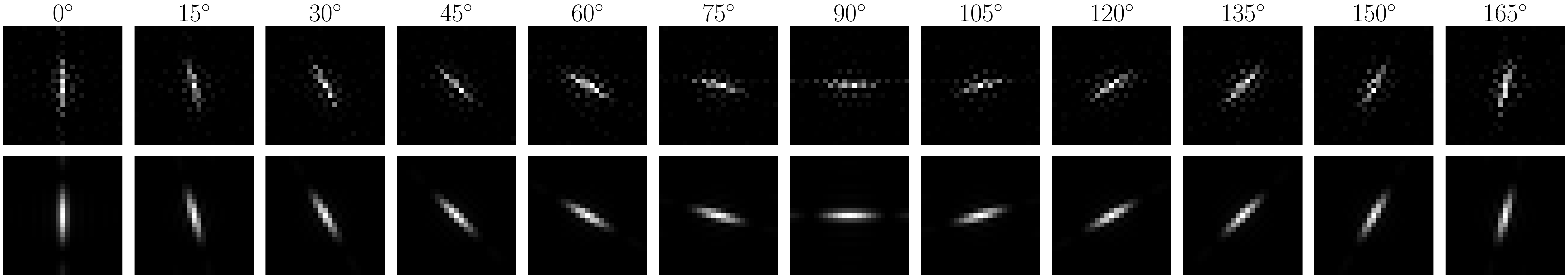}
  	\captionsetup{width=0.9\textwidth}
  	\caption{First row: $12$ estimated blur kernels; second row: the associated real (simulated) blur kernels.}
  	\label{figure4.2}
  \end{figure}
  \begin{figure}[!h]
  	\centering
  	\includegraphics[width=1.0\textwidth]{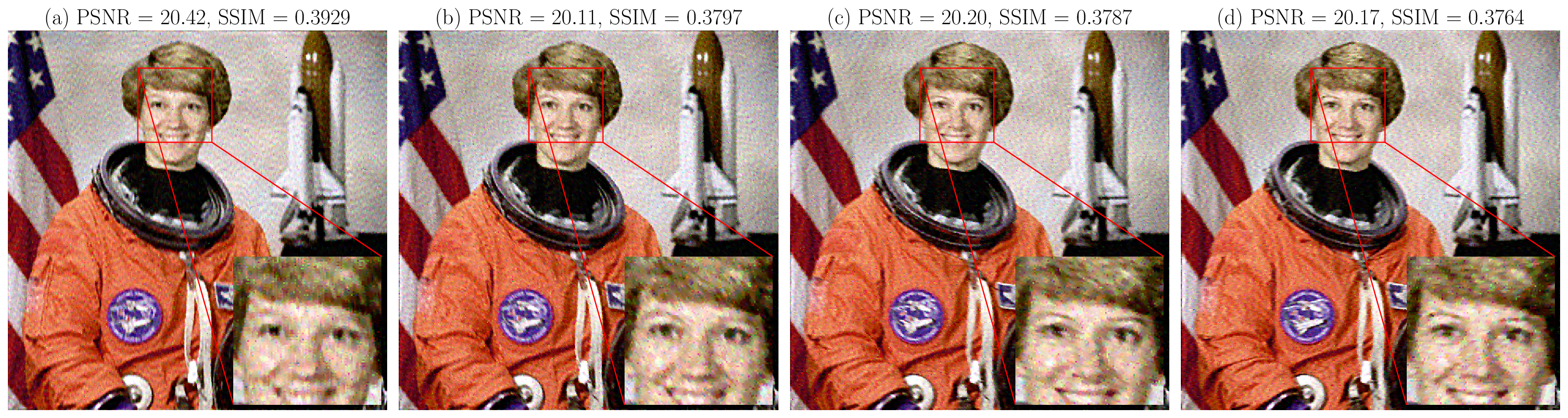}
  	\includegraphics[width=1.0\textwidth]{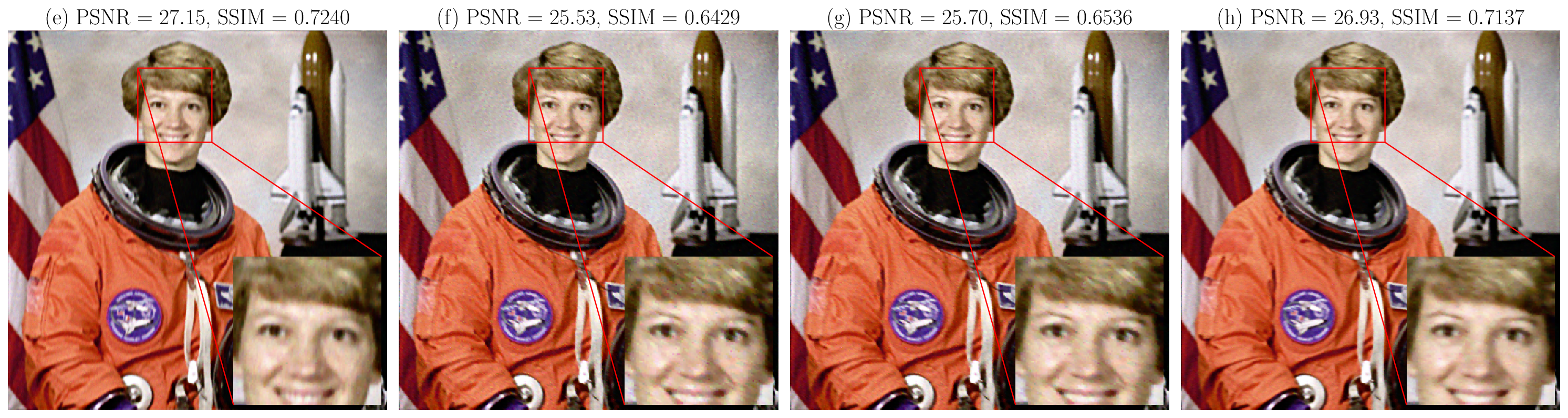}
  	\captionsetup{width=0.9\textwidth}
  	\caption{First row: $4$ deconvolved images; second row: the associated enhanced deconvolved images. The baseline image is the ground-truth sharp image $x$.  The PSNR (peak signal-to-noise ratio) and SSIM (structural similarity) values of each showed image are calculated with respect to the ground-truth image $x$.}
  	\label{figure4.3}
  \end{figure}
  
  The manifold fitting leverages information from $\{\widetilde{x}_j: ||\widetilde{x_j} - \widetilde{x_i}||_{\mathrm{F}} \leq r_1 \ \text{and} \ j \neq i\}$ to enhance the smoothness of $\widetilde{x_i}$. Specifically, the goal of this step is to denoise $\widetilde{x_i}$ by projecting it onto the latent manifold. The parameters in the manifold fitting are set as $r_1 = 108$ and $r_2 = 10 \times 108$, ensuring that each $\widetilde{x}_i$ has five neighbours at least within its local neighbourhood. As illustrated in each column of Figure \ref{figure4.3}, the enhanced deconvolved image $\widehat{x}_i^{\ast}$ appears significantly smoother in most local regions than the corresponding $\widetilde{x}_i$. Two objective metrics, PSNR (peak signal-to-noise ratio) and SSIM (structural similarity), are used to evaluate the effectiveness of enhancement through manifold fitting. The PSNRs of the four deconvolved images in Figure \ref{figure4.3} are improved by $32.96\%$, $26.95\%$, $27.22\%$, and $33.51\%$, respectively, while the SSIMs improve by $84.27\%$, $69.31\%$, $72.59\%$, and $89.61\%$. In Figure \ref{figure4.4}, these two box plots suggest that the manifold fitting does refine the deconvolved images. The effectiveness of manifold fitting can be also test by comparing those convolved images. As defined in equation (\ref{eq3.8}), the convolution is reintroduced onto $\widehat{x}_i^{\ast}$. Based on the simulated blur kernels $\{k_i\}_{i=1}^n$, we define the ground-truth convolved images $\{\mathcal{Y}_i\}_{i=1}^n$ as
  \begin{equation}
  	\mathcal{Y}_i = k_i \ast x, \ \text{for} \ i = 1,\cdots,n.
  	\label{eq4.1}
  \end{equation}
  In Figure \ref{figure4.6}, it can be seen that the noise in the enhanced convolved image $\widetilde{y}_i$ is significantly reduced while the structure of $\mathcal{Y}_i$ is sufficiently recovered. 
  \begin{figure}[!h]
  	\centering
  	\includegraphics[width=1.0\textwidth]{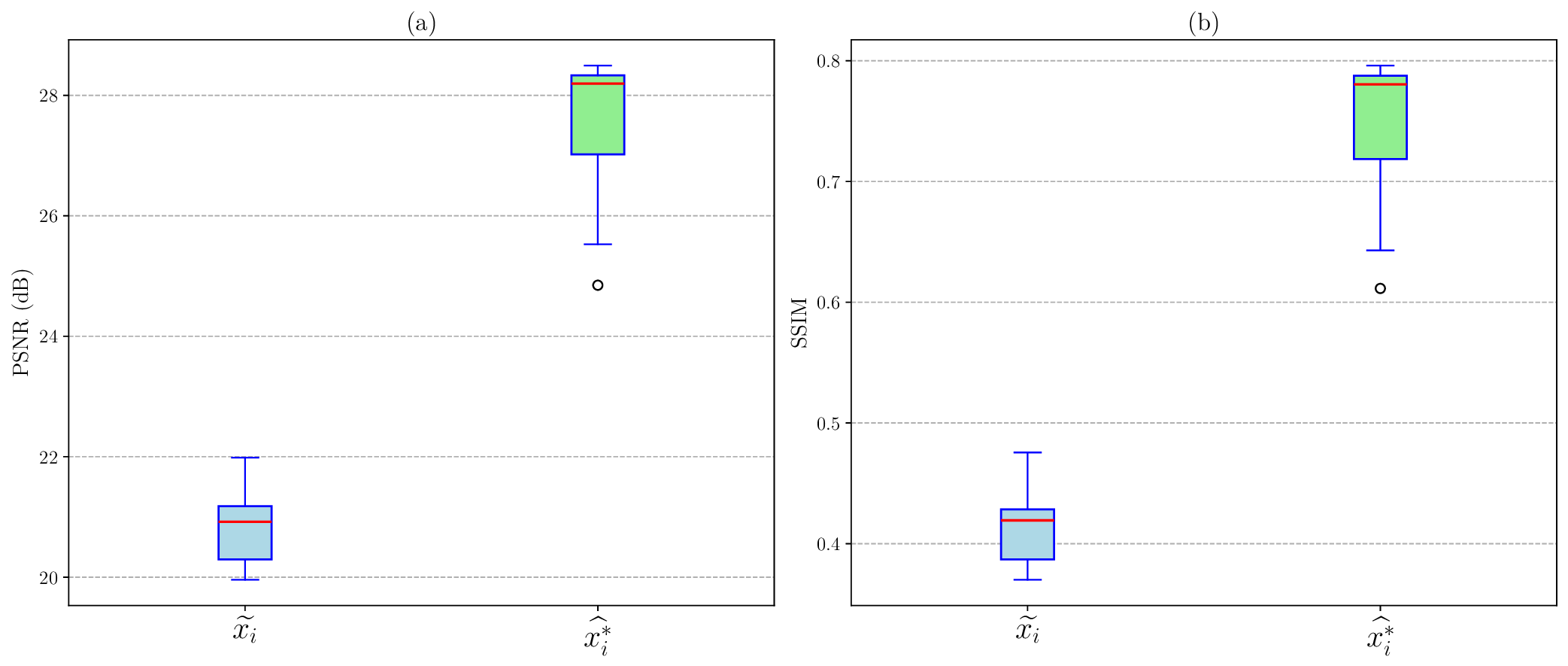}
  	\captionsetup{width=0.9\textwidth}
  	\caption{
(a) Comparison of PSNR values of deconvolved images $\{\widetilde{x}_i\}_{i=1}^n$ and those of enhanced (denoised) deconvolved images $\{\widehat{x}_i^{\ast}\}_{i=1}^n$. 
(b) Comparison of SSIM values of deconvolved images $\{\widetilde{x}_i\}_{i=1}^n$ and those of enhanced (denoised) deconvolved images $\{\widehat{x}_i^{\ast}\}_{i=1}^n$. The baseline image is the ground-truth sharp image $x$, i.e., the PSNR (peak signal-to-noise ratio) and SSIM (structural similarity) values of each $\widetilde{x}_i$ and $\widehat{x}_i^{\ast}$ are calculated with respect to the ground-truth image $x$.
}
 	\label{figure4.4}
  \end{figure}
  
  A single enhanced deconvolved image $\widehat{x}^{\ast}_i$ does not fully capture all local features, such as the astronaut's teeth (see the second row of Figure \ref{figure4.3}), and the effects of convolution and noise are not completely mitigated. Consequently, a single $\widehat{x}_i^{\ast}$ is not a fully satisfactory estimate of the sharp image $x$, though it approximates the latent image $x$ sufficiently in terms of PSNR and SSIM (as illustrated in Figure \ref{figure4.4}). Therefore, it is essential to consider the entire set $\{\widehat{x}_i^{\ast}\}_{i=1}^n$. By reintroducing convolution, enhanced convolved images $\{\widetilde{y}_i\}_{i=1}^n$ are generated using $\{\widehat{k}_i\}_{i=1}^n$ (as shown in equation (\ref{eq3.8})). Figure \ref{figure4.5} illustrates this process and compares an enhanced convolved image (Figure \ref{figure4.5}(d)) with its corresponding ground-truth convolved version (Figure \ref{figure4.1}(c)). One can observe that this $\widetilde{y}_i$ (Figure \ref{figure4.5} (d)) is significantly smoother than its original version $y_i$ (Figure \ref{figure4.5} (a)), with most noise eliminated.
  \begin{figure}[!h]
  	\centering
  	\includegraphics[width=1.0\textwidth]{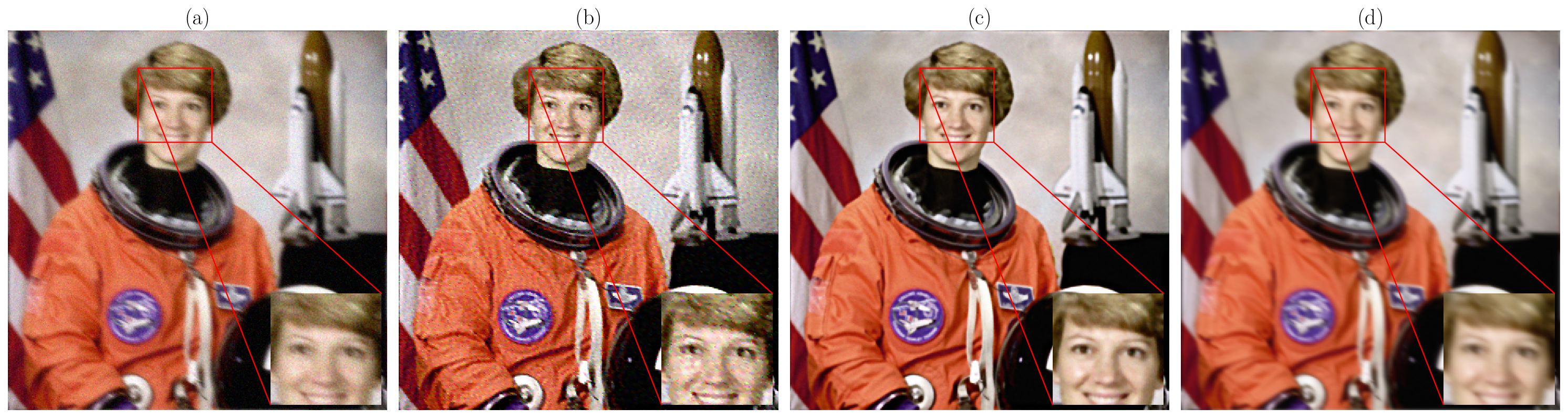}
  	\captionsetup{width=0.9\textwidth}
  	\caption{Process of denoising. (a) One blurred image $y_i$; (b) the associated deconvolved image $\widehat{x}_i$; (c) the associated enhanced deconvolved image $\widehat{x}^{\ast}_i$; (d) the associated enhanced convolved image $\widetilde{y_i}$.}
  	\label{figure4.5}
  \end{figure}
  
  Manifold fitting is also applied directly to the blurred images $\{y_i\}_{i=1}^n$, producing other enhanced convolved images denoted as $\ddot{y_1}, \cdots, \ddot{y_n}$, with tuning parameters $r_1 = 61$ and $r_2 = 10 \times 61$. This above values of $r_1$ and $r_2$ are chosen to maintain a comparable number of neighbouring images for each sample, consistent with the application of manifold fitting to the deconvolved images $\{\widehat{x_i}\}_{i=1}^n$. Nevertheless, as shown in Figure \ref{figure4.6}, the improvement achieved by applying manifold fitting directly to $\{y_i\}_{i=1}^n$ is less pronounced than that obtained for $\{\widehat{x_i}\}_{i=1}^n$. Applying manifold fitting to $\{\widehat{x}_i\}_{i=1}^n$ effectively assumes that the latent manifold represents the ground-truth image $x$, whereas applying it to $\{y_i\}_{i=1}^n$ assumes the latent manifold corresponds to the ground-truth convolved images $\{\mathcal{Y}_i\}_{i=1}^n$. Since the effect of convolution in $\{y_i\}_{i=1}^n$ is mitigated by solving equation (\ref{eq3.8}), $\{\widehat{x}_i\}_{i=1}^n$ complement each other in the source space and can be assembled to give an estimate of $x$. The essence of manifold fitting is to compute the weighted average of samples in the neighbourhood of any given sample, thereby significantly reducing noise in $\{\widehat{x}_i\}_{i=1}^n$. However, when $n$ is small, the blur kernel angles associated with each $y_i$ exhibit substantial variation, resulting in only a limited number of neighbouring blurred images for each $y_i$, often with blur kernels that differ widely in angle. This variability hinders the ability of manifold fitting to distinguish between the effects of convolution and noise. Consequently, the manifold fitting on $\{\widehat{x}_i\}_{i=1}^n$ performs better than the direct application on $\{y_i\}_{i=1}^n$ directly, particularly when sample sizes $n$ are limited. As $n$ increases, the blurred images within any neighbourhood will become approximate identically distributed due to the similar blur kernels. Consequently, applying manifold fitting to $\{y_i\}_{i=1}^n$ will yield results comparable to those obtained by applying it to $\{\widehat{x_i}\}_{i=1}^n$.
  \begin{figure}[!h]
  	\centering
  	\includegraphics[width=1.0\textwidth]{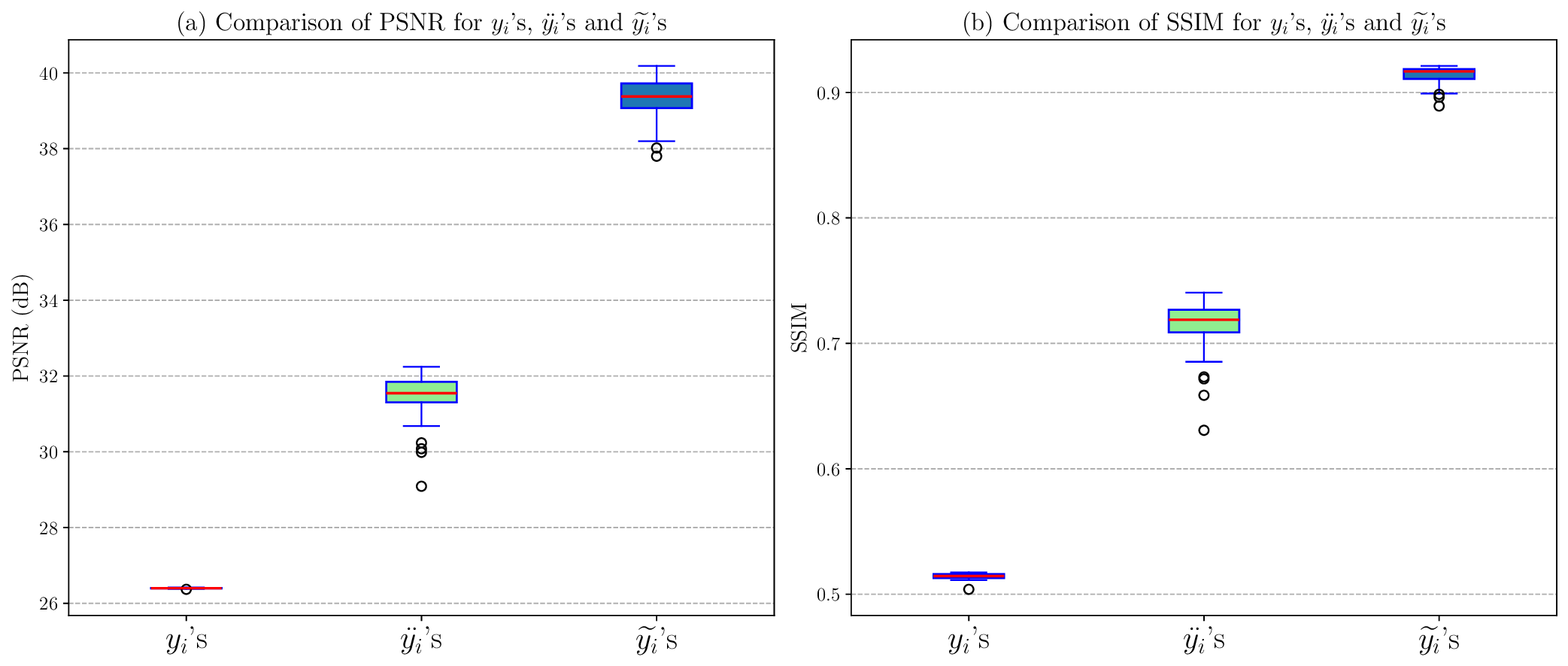}
  	\captionsetup{width=0.9\textwidth}
  	\caption{Comparisons of blurred images $\{y_i\}_{i=1}^n$, enhanced images(by applying manifold fitting on the blurred images directly) $\{\ddot{y}_i\}_{i=1}^n$'s and enhanced convolved images $\{\widetilde{y}_i\}_{i=1}^n$. The baseline images are $\{\mathcal{Y}_i\}_{i=1}^n$, i.e., the PSNR (peak signal-to-noise ratio) and SSIM (structural similarity) values of each image here are calculated with respect to the associated ground-truth convolved image $\mathcal{Y}_i$.}
  	\label{figure4.6}
  \end{figure}
  
  \subsection{Simulation results and discussion}
  In our final stage, we employ non-blind deconvolution (i.e., solving equation (\ref{eq3.8})) to reconstruct the sharp image (see Figure \ref{figure4.7} (d)), achieving a PSNR of $28.89$ dB and an SSIM of $0.7869$. The tuning parameters are set as follows: $\lambda_1=2000$ and $\alpha=0.80$. Our proposed method significantly enhances the quality of the reconstructed image, yielding a $13.34\%$ improvement in PSNR and a $36.19\%$ improvement in SSIM over conventional blind deconvolution (see Figure \ref{figure4.7} (a)). The noise visible in Figure \ref{figure4.7} (a) originates from the blurred images $\{y_i\}_{i=1}^n$. However, by mitigating this noise through manifold fitting, our method produces a better result in Figure \ref{figure4.7} (d), where local regions exhibit reduced noise and smoother texture. This outcome demonstrates that our approach effectively recovers a high-quality sharp image with more accurately estimated pixel intensities and preserved structural details, based on the enhanced convolved images $\{\widetilde{y}_i\}_{i=1}^n$. 
  
  Additionally, we reconstruct the latent image (Figure \ref{figure4.7} (b)) from an alternative set of enhanced deconvolved images $\{\ddot{y}_i\}_{i=1}^n$, which are produced by directly applying manifold fitting to $\{y_i\}_{i=1}^n$. This approach yields moderate improvements, with increases of $6.00\%$ in PSNR and $11.65\%$ in SSIM. As shown in the magnified local area (the face), though noise in Figure \ref{figure4.7} (b) is substantially reduced compared to Figure \ref{figure4.7} (a), the image appears nonsmoothed and blurred. In contrast, our method enhances the deconvolved images, ensuring that manifold fitting is based on more closely approximated identically distributed samples, thereby producing more reliable enhanced convolved images. Consequently, the reconstructed image (Figure \ref{figure4.7} (b)) achieves superior quality compared to Figure \ref{figure4.7} (d). The sharp image Figure \ref{figure4.7} (c) is derived by minimising $\sum\limits_{i=1}^n ||\widetilde{y}_i - \widehat{k}_i \ast x||_\mathrm{F}^2$. A comparison between Figures \ref{figure4.7} (c) and \ref{figure4.7} (d) highlights that including the prior in equation (\ref{eq3.4}) can improve PSNR and SSIM further. The prior term, $\log p_X(x)$, effectively stabilises the solution, suppresses noise, and preserves essential structural details, which is characteristic of both blind and non-blind deconvolution techniques.
  \begin{figure}[!h]
  	\centering
  	\includegraphics[width=1.0\textwidth]{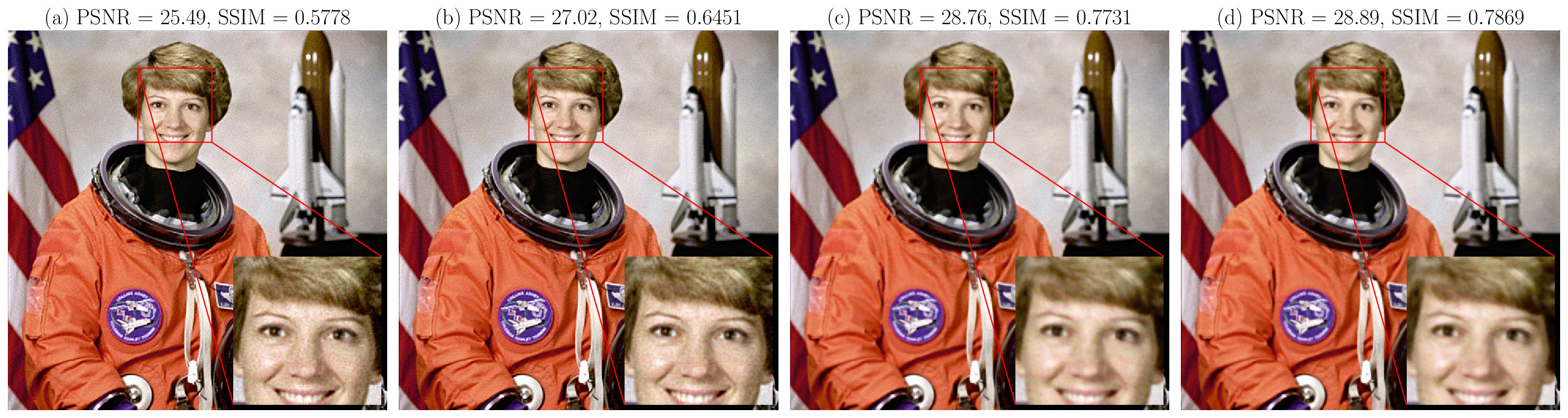}
  	\captionsetup{width=0.9\textwidth}
  	\caption{Reconstructed sharp images. (a) Zhou et al. (2021) (input images: $\{y_i\}_{i=1}^n$); (b) Zhou et al. (2021) (input images: $\{\ddot{y}_i\}_{i=1}^n$); (c) simple multi-frame blind deconvolution, $\mathcal{F}^{-1} \bigg(\frac{\sum\limits_{i=1}^n \overline{\mathcal{F}(\widetilde{k_i})} \ast \mathcal{F}(\widetilde{y}_i)}{\sum\limits_{i=1}^n \overline{\mathcal{F}(\widetilde{k}_i)} \ast \mathcal{F}(\widetilde{k}_i)} \bigg)$; (d) proposed. The baseline image is the ground-truth sharp image $x$, i.e., the PSNR (peak signal-to-noise ratio) and SSIM (structural similarity) values of each reconstructed sharp image are calculated with respect to the ground-truth image $x$.}
  	\label{figure4.7}
  \end{figure}
  
  We also assess the effect of the tuning parameter $r_1$ used in manifold fitting on reconstructing latent image for the above two methods, taking PSNR and SSIM as evaluation metrics. The PSNR curves reveal significant fluctuations for smaller values of $r_1$ (from $90$ to $105$). Compared to the method without reintroducing the convolution (i.e., reconstruction via equation (\ref{eq2.12})), our proposed method consistently achieves higher PSNR values, peaking at $r_1=110$. Beyond this point, PSNR declines for both methods. For SSIM, the proposed method similarly outperforms the alternative. Initially, the SSIM curves for both methods fluctuate, then increase and stabilise around $0.8$. These results demonstrate that for the proposed method, the optimal value $r_1$ is near $110$  in terms of reconstruction quality. We further assess the sensitivity of PSNR and SSIM to the value of tuning parameter $r_2$, finding that for a fixed $r_1$, the performance of the proposed method is not sensitive to value changes of $r_2$. This indicates that choosing five neighbours, as defined by $r_1, r_2$ and $\widetilde{x}_i$, is sufficient for achieving an  optimal improvement by using our proposed method. 
  \begin{figure}[!h]
  	\centering
  	\includegraphics[width=1.0\textwidth]{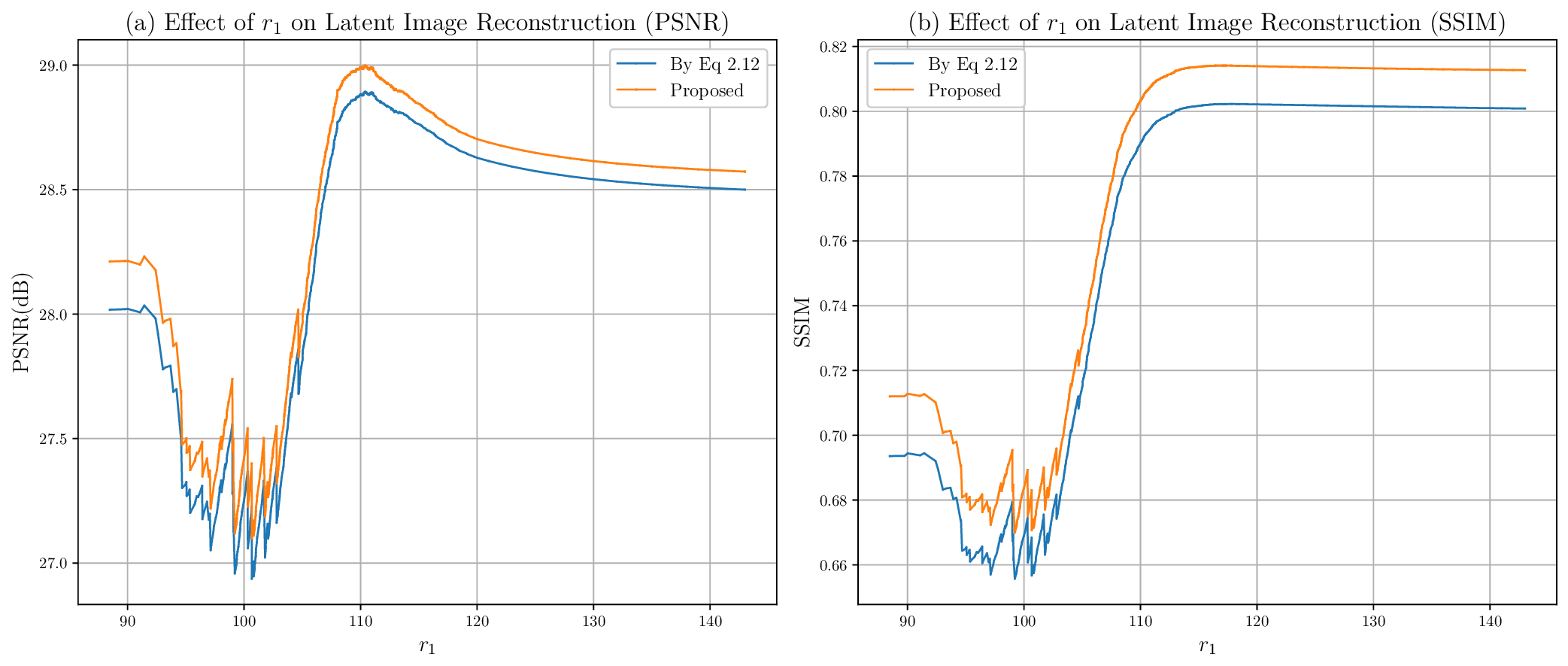}
  	\captionsetup{width=0.9\textwidth}
  	\caption{Comparison of the effect of $r_1$ on the latent image reconstruction quality, evaluated using two metrics: PSNR (peak signal-to-noise ratio) and SSIM (structural similarity):(a) Without prior (Equation 2.8). (b) With prior (Equation 2.9, proposed). }
  	\label{figure4.8}
  \end{figure}
  
\section{Conclusion}
\label{sec:conc}
 In this paper, we have hypothesised that images captured by a RSA imaging system can be modelled as noisy multi-frame convolutions of a latent sharp image with a sequence of PSFs and that the latent image is embedded into a low dimensional manifold within a high dimensional ambient space.
 Under this assumption,  we have developed a novel method to enhance images captured by the RSA imaging system and to reconstruct a sharper image in the framework of blind manifold deconvolution. In the proposed procedure, we first deconvolve RSA captured images by using the improved $\text{MAP}$ framework. We then fit a manifold to the deconvoluted images and enhance them by exploring their low dimensional manifold structures. Finally, we reintroduce a convolution operation on the enhanced images in order to produce denoised convoluted images for a sharp image reconstruction. Our simulation results have shown that the proposed method can outperform the conventional multi-frame blind deconvolution method in terms of estimating pixel intensities and preserving structural details. The results have also shown that fitting low dimensional manifolds to deconvolved acquired RSA images yields much better results than fitting manifolds directly to captured images. The proposed method can be further improved by incorporating recent advances in deep-learning-based deconvolution (Ren et al., 2020;Kotera et al.,2021) and in image fusion (Zhi et al., 2021; Sun et al., 2024). In this context, our proposed method offers a flexible image fusion framework centred on refined observed images. Deep learning approaches could be integrated in deblurring each captured image. Furthermore, deep learning could also be applied to reconstruct the final sharp image.

\section*{Funding}
This work of J.Z. was supported by the Engineering and Physical Sciences Research Council (EPSRC), United Kingdom,  under Grant EP/X038297/1.  The research of  D.L. is supported by the Chinese Scholarship Council (CSC)-University of Kent Scholarship.

\section*{Disclosure statement}
The authors report there are no competing interests to declare.

\section*{Acknowledgements}
 We are grateful to Professor Zhigang Yao, National University of Singapore for a fruitful discussion on his algorithm of  manifold fitting.

\appendix
\section{Appendix A: Convolution, DFT and Proof of Proposition \ref{proposition2.1}}
\label{Appendix A}
  \begin{proposition}
    \begin{equation}
  	 	\mathop{\arg\min}\limits_x \sum\limits_{i=1}^n ||y_i - k_i \ast x||_{\mathrm{F}}^2  = \mathcal{F}^{-1} \bigg(\frac{\sum\limits_{i=1}^n \overline{\mathcal{F}(\widetilde{k_i})} \odot \mathcal{F}(y_i)}{\sum\limits_{i=1}^n \overline{\mathcal{F}(\widetilde{k_i})} \odot \mathcal{F}(\widetilde{k_i})} \bigg) = \mathcal{F}^{-1} \bigg(\sum\limits_{i=1}^n \frac{\overline{\mathcal{F}(\widetilde{k_i})} \odot \mathcal{F}(\widetilde{k_i})}{\sum\limits_{j=1}^n \overline{\mathcal{F}(\widetilde{k_j})} \odot \mathcal{F}(\widetilde{k_j})} \odot \mathcal{F}(\widetilde{x_i}) \bigg),
  	 	\label{eq2.11}
  	\end{equation}
  	\label{proposition2.1}
  \end{proposition}
 where $\odot$ and $\frac{\mathbf{A}}{\mathbf{B}}$ denote the Hadamard (element-wise) product and division respectively, $\mathcal{F}(\cdot)$ is the 2-dimensional discrete Fourier transform (DFT), and $\mathcal{F}^{-1}(\cdot)$ is the 2-dimensional inverse discrete Fourier transform (IDFT). DFT and IDFT are both applied channel-wise. Here, $\widetilde{k_i}$ is the padded version of $\widehat{k_i}$, as detailed in Appendix \ref{Appendix A}. The term $\frac{\overline{\mathcal{F}(\widetilde{k_i})} \odot \mathcal{F}(\widetilde{k_i})}{\sum\limits_{j=1}^n \overline{\mathcal{F}(\widetilde{k_j})} \odot \mathcal{F}(\widetilde{k_j})}$ represents the normalised weight of each $\widetilde{x_i}$, where $\widetilde{x_i}$ can be seen as the deconvolution of the blurred image $y_i$ with its corresponding blur kernel $\widehat{k_i}$.

To facilitate the proof of Proposition \ref{proposition2.1}, we need to introduce more notations.
  For any arbitrary image $x \in \mathbb{R}^{C \times a \times b}$ , the image domain is defined as $\mathbb{A} := \{(c,w,h) \in \mathbb{N}^3: 0 \leq c \leq C-1, 0 \leq w \leq a-1, 0 \leq h \leq b-1 \}$. The $2$-dimensional DFT applied channel-wise is defined as:
  \begin{equation}
  	\mathcal{F}(x)(c,w,h) := \sum\limits_{u=0}^{a-1} \sum\limits_{v=0}^{b-1} x(c,u,v) \exp \big[-2 \pi \mathbf{i} (\frac{uw}{a} + \frac{vh}{b}) \big], \ \text{for} \ c=0,\cdots,C-1,
  	\label{eqA.1}
  \end{equation}
  where $\mathcal{F}(x) \in \mathbb{C}^{C \times a \times b}$ shares the same image domain $\mathbb{A}$ as $x$ and $\mathbf{i}$ represents the imaginary unit.

  The convolution theorem implies that the periodic convolution operation can be expressed as the Hadamard product in the frequency domain. Normally, the kernel size $s$ is assumed to be an odd number, sufficiently smaller than $a$ and $b$, such that $s=2s'+1$ and $s'$ is a positive integer. Each $k_i$, is indexed as $\{(w,h) \in \mathbb{Z}^2: -s' \leq w \leq s', -s' \leq h \leq s' \}$. In the case of $2$-dimensional discrete convolution, the \textit{periodic convolution} refers to the use of wrap-around (or circular) padding. This ensures that $k_i \otimes x$ has the same dimensions as $x$ and its each element is computed as:
  \begin{equation}
    (k_i \ast x) (c,w,h) := \sum\limits_{u=-s'}^{s'} \sum\limits_{v=-s'}^{s'} k_i(u,v) \cdot x \bigg(c, (w-u) \mod a, (h-v) \mod b \bigg),
    \label{eqA.2}
  \end{equation}
  where the image domain of $(k_i \ast x) \in \mathbb{R}^{C \times a \times b}$ is also $\mathbb{A}$. To use the convolution theorem, the kernel $k$ should also be padded so that its shape becomes $C \times a \times b$ and then be shifted. The conversion process is defined as:
  \begin{equation}
   \begin{gathered}
    \widetilde{k}_i(c,w,h) := 
    \begin{cases}
      k_i(w,h), \ \text{if} \ 0 \leq w \leq s' \ \text{and} \ 0 \leq h \leq s', \\
      k_i(w-a,h), \ \text{if} \ a-s' \leq w \leq a-1 \ \text{and} \ 0 \leq h \leq s', \\
      k_i(w,h-b), \ \text{if} \ 0 \leq w \leq s' \ \text{and} \ b-s' \leq h \leq b-1, \\
      k_i(w-a,h-b), \ \text{if} \ a-s' \leq w \leq a-1 \ \text{and} \ b-s' \leq h \leq b-1, \\
      0, \ \text{otherwise},
    \end{cases}
   \end{gathered}
   \label{eqA.3}
  \end{equation}
  for $c=0,\cdots,C-1$. Here, $\widetilde{k}_i \in \mathbb{R}^{C \times a \times b}$ is the converted kernel with its domain defined as $\mathbb{A}$.
    
  Then, the periodic convolution can be calculated based the convolution theorem:
  \begin{lemma}
    $k_i \ast x = \mathcal{F}^{-1}\bigg( \mathcal{F}(\widetilde{k_i}) \odot \mathcal{F}(x))\bigg)$.
    \label{LemmaA.1}
  \end{lemma}
  
  \begin{lemma}
    Parseval Equality: for any $\mathbf{X} \in \mathbb{R}^{C \times a \times b}$ having domain $\mathbb{A}$,
    \begin{equation}
     ||\mathbf{X}||_{\mathrm{F}}^2 := \sum\limits_{c=0}^{C-1}  \sum\limits_{w=0}^{a-1} \sum\limits_{h=0}^{b-1} |\mathbf{X}(c,w,h)|^2 = \sum\limits_{c=0}^{C-1} \sum\limits_{w=0}^{a-1} \sum\limits_{h=0}^{b-1} \frac{1}{ab} |\mathcal{F}(\mathbf{X})(c,w,h)|^2
    \label{eqA.4}
    \end{equation}
    \label{lemmaA.2}
    \end{lemma}
    
  The proof of Proposition \ref{proposition2.1} is provided:
  \begin{proof}
    Using Lemma\ref{lemmaA.2}, we have
    \begin{equation}
      \begin{split}
        L = \sum\limits_{i=1}^n ||y_i - k_i \ast x||_{\mathrm{F}}^2 &\propto \sum\limits_{i=1}^n \sum\limits_{c=0}^{C-1} \sum\limits_{w=0}^{a-1} \sum\limits_{h=0}^{b-1} |\mathcal{F}(y_i - k_i \ast x)(c,w,h)|^2 \\
        &= \sum\limits_{i=1}^n \sum\limits_{c=0}^{C-1} \sum\limits_{w=0}^{a-1} \sum\limits_{h=0}^{b-1} |\mathcal{F}(y_i)(c,w,h) - \mathcal{F}(\widetilde{k_i})(c,w,h) \cdotp \mathcal{F}(x)(c,w,h)|^2 \\
        &= \sum\limits_{i=1}^n \sum\limits_{c=0}^{C-1} \sum\limits_{w=0}^{a-1} \sum\limits_{h=0}^{b-1} \bigl\{|\mathcal{F}(y_i)(c,w,h)|^2 + |\mathcal{F}(\widetilde{k_i})(c,w,h)|^2 \cdotp |\mathcal{F}(x)(c,w,h)|^2 \\
        &- \overline{\mathcal{F}(y_i)(c,w,h)} \cdotp \mathcal{F}(\widetilde{k_i})(c,w,h) \cdotp \mathcal{F}(x)(c,w,h) \\
        &- \mathcal{F}(y_i)(c,w,h) \cdotp \overline{\mathcal{F}(\widetilde{k_i})(c,w,h)} \cdotp \overline{\mathcal{F}(x)(c,w,h)} \bigr\}
      \end{split}
      \label{eqA.5}
    \end{equation}
    We represent the real and imaginary parts of an arbitrary complex number $z$ as $\textbf{Re}(z)$ and $\textbf{Im}(z)$, respectively. We then calculate the first-order derivatives of $L$ in terms of $\textbf{Re}\big(\mathcal{F}(x)(c,w,h) \big)$ and $\textbf{Im}\big(\mathcal{F}(x)(c,w,h) \big)$:
    \begin{eqnarray}
      \frac{\partial L}{\partial \textbf{Re}\big(\mathcal{F}(x)(c,w,h) \big)}& =& 2\textbf{Re}\big(\mathcal{F}(x)(c,w,h) \big)\sum\limits_{i=1}^n |\mathcal{F}(\widetilde{k_i})(c,w,h)|^2\nonumber \\
    &&- 2\sum\limits_{i=1}^n \textbf{Re}\big(\overline{\mathcal{F}(\widetilde{k_i})(c,w,h)} \cdotp \mathcal{F}(y_i)(c,w,h) \big)
      \label{eqA.6}
     \end{eqnarray}
     \begin{eqnarray}
       \frac{\partial L}{\partial \textbf{Im}\big(\mathcal{F}(x)(c,w,h) \big)} &=& 2\textbf{Im}\big(\mathcal{F}(x)(c,w,h) \big)\sum\limits_{i=1}^n |\mathcal{F}(\widetilde{k_i})(c,w,h)|^2\nonumber\\
& &- 2\sum\limits_{i=1}^n \textbf{Im}\big(\overline{\mathcal{F}(\widetilde{k_i})(c,w,h)} \cdotp \mathcal{F}(y_i)(c,w,h) \big)
       \label{eqA.7}
     \end{eqnarray}
    Setting these two first-order derivatives to zero, we obtain:
    \begin{equation}
      \begin{gathered}
        \begin{cases}
          \textbf{Re}\big(\mathcal{F}(\widehat{x})(c,w,h) \big) = \frac{\sum\limits_{i=1}^n \textbf{Re}\big(\overline{\mathcal{F}(\widetilde{k_i})(c,w,h)} \cdotp \mathcal{F}(y_i)(c,w,h) \big)}{\sum\limits_{i=1}^n |\mathcal{F}(\widetilde{k_i})(c,w,h)|^2} \\
          \textbf{Im}\big(\mathcal{F}(\widehat{x})(c,w,h) \big) = \frac{\sum\limits_{i=1}^n \textbf{Im}\big(\overline{\mathcal{F}(\widetilde{k_i})(c,w,h)} \cdotp \mathcal{F}(y_i)(c,w,h) \big)}{\sum\limits_{i=1}^n |\mathcal{F}(\widetilde{k_i})(c,w,h)|^2}.
        \end{cases}
      \end{gathered}
      \label{eqA.8}	
    \end{equation}
    Combining the real and imaginary parts of $\mathcal{F}(\widehat{x})(c,w,h)$ and considering all $(c,w,h)$, the estimated image $\widehat{x}$ is reconstructed by the $2$-dimensional Inverse Discrete Fourier transform (IDFT):
    \begin{equation}
      \widehat{x} = \mathcal{F}^{-1} \bigg(\frac{\sum\limits_{i=1}^n \overline{\mathcal{F}(\widetilde{k_i})} \odot \mathcal{F}(y_i)}{\sum\limits_{i=1}^n \overline{\mathcal{F}(\widetilde{k_i})} \odot \mathcal{F}(\widetilde{k_i})} \bigg).
      \label{eqA.9}
    \end{equation}
  \end{proof}

\section{Appendix B: Half-quadratic Splitting}
\label{Appendix B}
  After introducing two auxiliary variables $\gamma_h \in \mathbb{R}^{C \times b \times a}$ and $\gamma_v \in \mathbb{R}^{C \times b \times a}$, equation (\ref{eqxalp}) is modified to  
  \begin{equation}
    (\widehat{x}, \widehat{\gamma_h}, \widehat{\gamma_v}) = \mathop{\arg\min}\limits_{x,\gamma_h,\gamma_v} \bigg(\lambda \sum\limits_{i=1}^{n} ||y_i - \widehat{k_i} \ast x||_{\mathrm{F}}^2 + \sum\limits_{j \in \{h,v\}} ||\textbf{vec}(\gamma_j)||_{\alpha}^{\alpha} + \frac{\beta}{2} \sum\limits_{j \in \{h,v\}} ||\gamma_j - G_j \ast x||_{\mathrm{F}}^2 \bigg).
    \label{eqB.1}
  \end{equation}
   Under the framework of half-quadratic splitting, the solution of equation (\ref{eqB.1}) converges to that of equation (\ref{eqxalp}) as $\beta \rightarrow \infty$. Equation (\ref{eqB.1}) is divided into two sub-problems:
  \begin{equation}
    \widehat{x} = \mathop{\arg\min}\limits_{x,\gamma_h,\gamma_v} \bigg(\lambda \sum\limits_{i=1}^{n} ||y_i - \widehat{k_i} \ast x||_{\mathrm{F}}^2 + \frac{\beta}{2} \sum\limits_{j \in \{h,v\}} ||\gamma_j - G_j \ast x||_{\mathrm{F}}^2 \bigg), \ \text{and}
    \label{eqB.2}
  \end{equation}
  \begin{equation}
    \widehat{\gamma_j} = \mathop{\arg\min}\limits_{\gamma_j} \bigg(||\textbf{vec}(\gamma_j)||_{\alpha}^{\alpha} + \frac{\beta}{2} ||\gamma_j - f_j \ast x||_{\mathrm{F}}^2 \bigg) \ \text{for} \ j=h,v.
    \label{eqB.3}
  \end{equation}
  Using Proposition \ref{proposition2.1}, an approximation solution of equation (\ref{eqB.2}) is given by:
  \begin{equation}
    \widehat{x} = \mathcal{F}^{-1} \bigg(\frac{2\lambda \sum\limits_{i=1}^n \overline{\mathcal{F}(\widetilde{k_i})} \odot \mathcal{F}(y_i) + \beta \sum\limits_{j \in \{h,v\}} \overline{\mathcal{F}(\widetilde{G_j})} \odot \mathcal{F}(\gamma_j)}{2\lambda \sum\limits_{i=1}^n \overline{\mathcal{F}(\widetilde{k_i})} \odot \mathcal{F}(\widetilde{k_i}) + \beta \sum\limits_{j \in \{h,v\}} \overline{\mathcal{F}(\widetilde{G_j})} \odot \mathcal{F}(\widetilde{G_j})} \bigg),
    \label{eqB.4}
  \end{equation}
  where $\widetilde{G_j}$ is the padded version of $G_j$ using equation (\ref{eqA.3}).
  Using the IRLS algorithm, the approximate solution of equation (\ref{eqB.3}) is given by:
  \begin{equation}
    \gamma_j^t = \frac{\beta (G_j \ast x)}{\beta + \alpha [(\gamma_j^{t-1} \odot \gamma_j^{t-1}) \oplus \epsilon]^{\frac{\alpha}{2}-1}}.
    \label{eqB.5}
  \end{equation}
  Here, $\oplus$ is defined as the element-wise addition and the exponent ($\frac{p}{2}-1$) is applied element-wise. The full algorithm for solving equation (\ref{eqB.1}) is summarised in Figure \ref{flowchart1}. In our simulation, the relevant parameters are set as follows: $\beta_{\text{ini}}=1.0, \beta_{\text{max}}=10^{20}, r=2, T=5, \epsilon=10^{-9}$.
  
  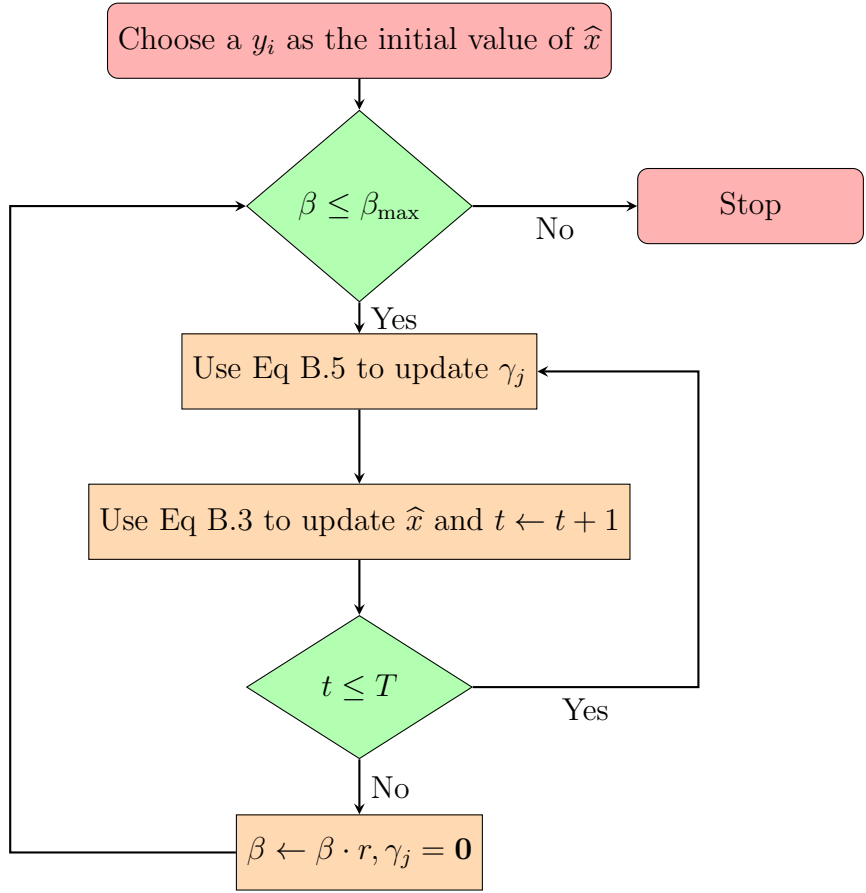
\begin{figure}[!h]
    \centering
    \begin{tikzpicture}[node distance=1.2cm][!h]
      \node (start) [startstop] {Choose a $y_i$ as the initial value of $\widehat{x}$};
      \node (decision1) [decision, below of=start, yshift=-1cm] {$\beta \leq \beta_{\text{max}}$};
      \node (process1) [process, below of=decision1, yshift=-1cm] {Use Eq \ref{eqB.5} to update $\gamma_j$};
      \node (process2) [process, below of=process1, yshift=-0.8cm] {Use Eq \ref{eqB.3} to update $\widehat{x}$ and $t \leftarrow t+1$};
      \node (decision2) [decision, below of=process2, yshift=-1cm] {$t \leq T$};
      \node (process3) [process, below of=decision2, yshift=-1cm] {$\beta \leftarrow \beta \cdot r, \gamma_j = \mathbf{0}$};
      \node (stop) [startstop, right of=decision1, xshift=4cm] {Stop};
            
      \draw [arrow] (start) -- (decision1);
      \draw [arrow] (decision1) -- node[anchor=west] {Yes} (process1);
      \draw [arrow] (decision1) -- node[anchor=north] {No} (stop);
      \draw [arrow] (process1) -- (process2);
      \draw [arrow] (process2) -- (decision2);
      \draw [arrow] (decision2) -- node[anchor=west] {No} (process3);
          
      \draw [arrow] (decision2.east) -| ([xshift=3cm] decision2.east) node[pos=0.25, below] {Yes} |- (process1.east);
      \draw [arrow] (process3.west) -| ([xshift=-3cm] process3.west) |- (decision1.west);
      \label{flowchat1}
    \end{tikzpicture}
    \captionsetup{width=0.9\textwidth}
    \caption{The flowchart of half-quadratic splitting framework.}
    \label{flowchart1}
  \end{figure}

\section{Appendix C: Converting the optimisation problems in  equation \ref{eqQ} and equation (\ref{eqQ}) with $\nabla_jy_i$ and $\nabla_j\widehat{x}$ replaced by $y_i$ and $\widehat{x}$ into the problems of Lasso regression}
\label{Appendix C}
 We first consider the following optimisation problem similar to that in  equation \ref{eqQ}
  \begin{equation}
	\mathop{\arg\min}\limits_{q_i} \bigg\{||y_i - q_i \ast \widehat{x}||_{\mathrm{F}}^2 + \mu ||\textbf{vec}(q_i)||_1\bigg\},
	\label{eqC.1}
  \end{equation}
  where the objective function is denoted as $L$. 

 Assuming all elements of $q_i$ are non-negative, the $L_1$-norm $||\textbf{vec}(q_i)||_1$ becomes differentiable. Let $\widetilde{q_i}$ denote the padded version of $q_i$, with padding defined in Eq \ref{eqA.3}. Given $\widetilde{q_i}(0,h,w) = \widetilde{q_i}(1,h,w) = \cdots = \widetilde{q_i}(C-1,h,w)$ for all $h,w$, we define $\mathring{q_i}$ as
  \begin{equation}
	\mathring{q_i}(h,w) = \widetilde{q_i}(c,h,w), \ \text{for} \ h=0,\cdots,a-1 \ \text{and} \ w =0,\cdots,b-1. 
	\label{eqC.2}
  \end{equation} 
  The first term in the objective function can then be expressed as:
  \begin{equation}
	\begin{split}
	  ||y_i - q_i \ast \widehat{x}||_{\mathrm{F}}^2 &= \frac{1}{ab} \sum\limits_{c=0}^{C-1} \sum\limits_{w=0}^{a-1} \sum\limits_{h=0}^{b-1} \big\{|\mathcal{F}(y_i)(c,w,h)|^2 + |\mathcal{F}(\mathring{q_i})(w,h)|^2 \cdotp |\mathcal{F}(\widehat{x})(c,w,h)|^2 \\
	  &- \overline{\mathcal{F}(y_i)(c,w,h)} \cdotp \mathcal{F}(\mathring{q_i})(w,h) \cdotp \mathcal{F}(\widehat{x})(c,w,h) \\
	  &- \mathcal{F}(y_i)(c,w,h) \cdotp \overline{\mathcal{F}(\mathring{q_i})(w,h)} \cdotp \overline{\mathcal{F}(\widehat{x})(c,w,h)} \big\}.
	\end{split}
	\label{eqC.3}
  \end{equation}
  Based on the 2-dimensional DFT $\mathcal{F}(\mathring{q_i})(w,h) := \sum\limits_{u=0}^{a-1} \sum\limits_{v=0}^{b-1} \mathring{q_i}(u,v) \exp \big[-2 \pi \mathbf{i} (\frac{uw}{a} + \frac{vh}{b}) \big]$ and the Euler's formula $\mathbf{e}^{\mathbf{i}x} = \cos x + \mathbf{i} \sin x$, the following derivatives are obtained:
  \begin{equation}
	\frac{\partial \mathcal{F}(\mathring{q_i})(w,h)}{\partial \mathring{q_i}(u,v)} = \cos [2\pi(\frac{uw}{a} + \frac{vh}{b})] - \mathbf{i} \sin [2\pi(\frac{uw}{a} + \frac{vh}{b})],
	\label{eqC.4}
  \end{equation}
  \begin{equation}
	\frac{\partial \overline{\mathcal{F}(\mathring{q_i})(w,h)}}{\partial \mathring{q_i}(u,v)} = \cos [2\pi(\frac{uw}{a} + \frac{vh}{b})] + \mathbf{i} \sin [2\pi(\frac{uw}{a} + \frac{vh}{b})],
	\label{eqC.5}
  \end{equation}
  \begin{equation}
	\begin{split}
	  \frac{\partial |\mathcal{F}(\mathring{q_i})(w,h)|^2}{\partial \mathring{q_i}(u,v)} &= \cos [2\pi(\frac{uw}{a} + \frac{vh}{b})] [\mathcal{F}(\mathring{q_i})(w,h) + \overline{\mathcal{F}(\mathring{q_i})(w,h)}] \\
	  &\ \ \ \ + \mathbf{i} \sin [2\pi(\frac{uw}{a} + \frac{vh}{b})] [\mathcal{F}(\mathring{q_i})(w,h) - \overline{\mathcal{F}(\mathring{q_i})(c,w,h)}] \\
	  &=2 \cos [2\pi(\frac{uw}{a} + \frac{vh}{b})] \textbf{Re}[\mathcal{F}(\mathring{q_i})(w,h)] - 2 \sin [2\pi(\frac{uw}{a} + \frac{vh}{b})] \textbf{Im}[\mathcal{F}(\mathring{q_i})(w,h)] \\
	  &=2 \cos [2\pi(\frac{uw}{a} + \frac{vh}{b})] \bigg\{\sum\limits_{\theta=0}^{a-1}\sum\limits_{\eta=0}^{b-1} \mathring{q_i}(\theta,\eta) \cos [2\pi(\frac{\theta w}{a} + \frac{\eta h}{b})] \bigg\} \\
	  &\ \ \ \ +2 \sin [2\pi(\frac{uw}{a} + \frac{vh}{b})] \bigg\{\sum\limits_{\theta=0}^{a-1}\sum\limits_{\eta=0}^{b-1} \mathring{q_i}(\theta,\eta) \sin [2\pi(\frac{\theta w}{a} + \frac{\eta h}{b})] \bigg\} \\
	  &= 2 \sum\limits_{\theta=0}^{a-1}\sum\limits_{\eta=0}^{b-1} \mathring{q_i}(\theta,\eta) \cos \big[2\pi(\frac{uw}{a} + \frac{vh}{b}) - 2\pi(\frac{\theta w}{a} + \frac{\eta h}{b}) \big].
	\end{split}
	\label{eqC.6}
  \end{equation}
  The first-order partial derivative of the objective function $L$ with respect to $\mathring{q_i}(u,v)$ is expressed as:
  \begin{eqnarray}
	  \frac{\partial L}{\partial \mathring{q_i}(u,v)} &=& \frac{2}{ab} \sum\limits_{c=0}^{C-1} \sum\limits_{w=0}^{a-1} \sum\limits_{h=0}^{b-1} \left\{|\mathcal{F}(\widehat{x})(c,w,h)|^2 \cdotp \sum\limits_{\theta=0}^{a-1}\sum\limits_{\eta=0}^{b-1} \mathring{q_i}(\theta,\eta) \cos \left[2\pi(\frac{uw}{a} + \frac{vh}{b})\right.\right.\nonumber\\
&&\left.\left. - 2\pi(\frac{\theta w}{a} + \frac{\eta h}{b}) \right] 
	   - \textbf{Re}\left[\overline{\mathcal{F}(y_i)(c,w,h)} \cdotp \mathcal{F}(\widehat{x})(c,w,h) \right] \cdotp \cos \left[2\pi(\frac{uw}{a} + \frac{vh}{b}) \right]\right.\nonumber \\
	  & & \left.- \textbf{Im}\left[\overline{\mathcal{F}(y_i)(c,w,h)} \cdotp \mathcal{F}(\widehat{x})(c,w,h) \big] \cdotp \sin \big[2\pi(\frac{uw}{a} + \frac{vh}{b}) \right] \right\} + \mu\nonumber \\
	  &=& \frac{2}{ab} \sum\limits_{\theta=0}^{a-1}\sum\limits_{\eta=0}^{b-1} \mathring{q_i}(\theta,\eta)\nonumber\\
             & &\cdotp \left\{\sum\limits_{c=0}^{C-1} \sum\limits_{w=0}^{a-1} \sum\limits_{h=0}^{b-1} |\mathcal{F}(\widehat{x})(c,w,h)|^2 \cdotp \cos \left[2\pi(\frac{uw}{a} + \frac{vh}{b}) - 2\pi(\frac{\theta w}{a} + \frac{\eta h}{b}) \right] \right\}\nonumber \\
	  & & -\frac{2}{ab} \sum\limits_{c=0}^{C-1} \sum\limits_{w=0}^{a-1} \sum\limits_{h=0}^{b-1} \bigg\{\textbf{Re}\left[\overline{\mathcal{F}(y_i)(c,w,h)} \cdotp \mathcal{F}(\widehat{x})(c,w,h) \right] \cdotp \cos \left[2\pi(\frac{uw}{a} + \frac{vh}{b}) \right]\nonumber \\ 
	  & & \left.+ \textbf{Im}\left[\overline{\mathcal{F}(y_i)(c,w,h)} \cdotp \mathcal{F}(\widehat{x})(c,w,h) \right] \cdotp \sin \left[2\pi(\frac{uw}{a} + \frac{vh}{b}) \right] \right\} + \mu. 
	\label{eqC.7}
  \end{eqnarray}
  The second-order partial derivative of $L$ with respect to $\mathring{q_i}(u,v)$ is:
  \begin{equation}
	\begin{split}
	  \frac{\partial^2 L}{\partial \mathring{q_i}(u,v)^2} &= \frac{2}{ab} \bigg\{\sum\limits_{c=0}^{C-1} \sum\limits_{w=0}^{a-1} \sum\limits_{h=0}^{b-1} |\mathcal{F}(\widehat{x})(c,w,h)|^2 \cdotp \cos \big[2\pi(\frac{uw}{a} + \frac{vh}{b}) - 2\pi(\frac{uw}{a} + \frac{vh}{b}) \big] \bigg\} \\
	  &= 2 ||\widehat{x}||_{\mathrm{F}}^2 > 0
	\end{split}
	\label{eqC.8}
  \end{equation}
  
  Finally, we define $\mathbf{A}: \{(g,m) \in \mathbb{N}^2: 0 \leq g \leq ab-1, 0 \leq m \leq ab-1 \} \rightarrow \mathbb{R}$ as:
  \begin{equation}
  	\mathbf{A}(g,m) = \frac{1}{ab} \sum\limits_{c=0}^{C-1} \sum\limits_{w=0}^{a-1} \sum\limits_{h=0}^{b-1} |\mathcal{F}(\widehat{x})(c,w,h)|^2 \cdotp \cos \big[2\pi(\frac{uw}{a} + \frac{vh}{b}) - 2\pi(\frac{\theta w}{a} + \frac{\eta h}{b}) \big],
  	\label{eqC.9}
  \end{equation}
  where $g = ub + v$ and $m = \theta b + \eta$.
  
  \begin{lemma}
  	The matrix $\mathbf{A}$ defined as equation (\ref{eqC.9}) is symmetric and positive-definite.
  \end{lemma}
  \begin{proof}
  	\textbf{Symmetry} Using the identity $\cos \big[2\pi(\frac{uw}{a} + \frac{vh}{b}) - 2\pi(\frac{\theta w}{a} + \frac{\eta h}{b}) \big] = \cos \big[2\pi \frac{w}{a}(u-\theta) + 2\pi \frac{h}{b}(v-\eta) \big] = \cos \big[2\pi \frac{w}{a}(\theta - u) + 2\pi \frac{h}{b}(\eta - v) \big]$, we have $\mathbf{A}(g,m) = \mathbf{A}(m,g)$. Thus, $\mathbf{A}$ is symmetric. \\
  	
  	\textbf{Positive definiteness} 
  	For any non-zero vector $\mathbf{r} = (r_0,\cdots,r_{ab-1}) \in \mathbb{R}^{ab}$, consider the quadratic form $\mathbf{r}^{\textbf{T}} \mathbf{A} \mathbf{r}$:
  	\begin{equation}
  	  \begin{split}
  	  	\mathbf{r}^{\textbf{T}} \mathbf{A} \mathbf{r} &= \sum\limits_{g=0}^{ab-1} \sum\limits_{m=0}^{ab-1}  r_g r_m \mathbf{A}(g,m) \\
  	  	&=\frac{1}{ab} \sum\limits_{g=0}^{ab-1} \sum\limits_{m=0}^{ab-1} \sum\limits_{c=0}^{C-1} \sum\limits_{w=0}^{a-1} \sum\limits_{h=0}^{b-1} r_g r_m |\mathcal{F}(\widehat{x})(c,w,h)|^2 \cdotp \cos \big[2\pi(\frac{uw}{a} + \frac{vh}{b}) - 2\pi(\frac{\theta w}{a} + \frac{\eta h}{b}) \big] \\
  	  	&=\frac{1}{ab} \sum\limits_{c=0}^{C-1} \sum\limits_{w=0}^{a-1} \sum\limits_{h=0}^{b-1} |\mathcal{F}(\widehat{x})(c,w,h)|^2 \bigg\{\sum\limits_{g=0}^{ab-1} \sum\limits_{m=0}^{ab-1} r_g r_m \cos \big[2\pi(\frac{uw}{a} + \frac{vh}{b}) - 2\pi(\frac{\theta w}{a} + \frac{\eta h}{b})\bigg\}. 
  	  \end{split}
  	  \label{eqC.10}
  	\end{equation}
  	We denote $2\pi(\frac{uw}{a} + \frac{vh}{b})$ as $\Phi_{w,h}(g)$. Then,
  	\begin{equation}
  	  \begin{split}
  	  	&\sum\limits_{g=0}^{ab-1} \sum\limits_{m=0}^{ab-1} r_g r_m \cos \big[2\pi(\frac{uw}{a} + \frac{vh}{b}) - 2\pi(\frac{\theta w}{a} + \frac{\eta h}{b}) \big] = \sum\limits_{g=0}^{ab-1} \sum\limits_{m=0}^{ab-1} r_g r_m \cos \big[\Phi_{w,h}(g) - \Phi_{w,h}(m) \big] \\
  	  	&= \sum\limits_{g=0}^{ab-1} \sum\limits_{m=0}^{ab-1} r_g r_m \cos\Phi_{w,h}(g) \cos\Phi_{w,h}(m) + \sum\limits_{g=0}^{ab-1} \sum\limits_{m=0}^{ab-1} r_g r_m \sin\Phi_{w,h}(g) \sin\Phi_{w,h}(m) \\
  	  	&= \bigg[\sum\limits_{g=0}^{ab-1} r_g \cos\Phi_{w,h}(g) \bigg]^2 + \bigg[\sum\limits_{g=0}^{ab-1} r_g \sin\Phi_{w,h}(g) \bigg]^2 > 0
  	  \end{split}
  	  \label{eqC.11}
  	\end{equation}
  	Thus, the quadratic form $\mathbf{r}^{\textbf{T}} \mathbf{A} \mathbf{r}$ is strictly positive and the matrix $\mathbf{A}$ is positive-definite.
  \end{proof}
  
  Using the Cholesky decomposition, the symmetric and positive-definite matrix $\mathbf{A}$ can be expressed as $\mathbf{B}^{\textbf{T}} \mathbf{B}$ where $\mathbf{B}$ is an upper-triangular matrix with positive diagonal elements. The vectorisation of $\mathring{q_i}$ is denoted as $\textbf{vec}(\mathring{q_i}) = \big[\mathring{q_i}(0,0),\mathring{q_i}(0,1),\cdots,\mathring{q_i}(0,b-1),\mathring{q_i}(1,0),\mathring{q_i}(1,1),\cdots,\mathring{q_i}(1,b-1),\cdots,\mathring{q_i}(a-1,0),\mathring{q_i}(a-1,1),\cdots,\mathring{q_i}(a-1,b-1)\big]^{\textbf{T}}$. Based on this representation, Eq \ref{eqC.7} can be reformulated as follows:
  \begin{equation}
  	\frac{\partial L}{\partial \textbf{vec}(\mathring{q_i})} = 2 [\mathbf{B}^{\textbf{T}} \mathbf{B} \cdotp \textbf{vec}(\mathring{q_i}) - \mathbf{B}^{\textbf{T}} \mathbf{z}] + \mu \mathbf{1},
  	\label{eqC.12}
  \end{equation}
  where $\mathbf{B}^{\textbf{T}} \mathbf{z} = \mathbf{z}'$, 
  \begin{equation}
  	\begin{split}
  	  \mathbf{z}'(j) &= \frac{1}{ab} \sum\limits_{c=0}^{C-1} \sum\limits_{w=0}^{a-1} \sum\limits_{h=0}^{b-1} \bigg\{\textbf{Re}\big[\overline{\mathcal{F}(y_i)(c,w,h)} \cdotp \mathcal{F}(\widehat{x})(c,w,h) \big] \cdotp \cos \big[2\pi(\frac{uw}{a} + \frac{vh}{b}) \big] \\
  	  &\ \ \ \ \ + \textbf{Im}\big[\overline{\mathcal{F}(y_i)(c,w,h)} \cdotp \mathcal{F}(\widehat{x})(c,w,h) \big] \cdotp \sin \big[2\pi(\frac{uw}{a} + \frac{vh}{b}) \big] \bigg\}, \ \text{for} \ j=0,\cdots,ab-1,
  	\end{split}
  	\notag
  \end{equation}
  and $\mathbf{1}=[1,\cdots,1]^{\textbf{T}}$. Thus, the optimisation problem equation (\ref{eqC.1}) can be rewritten as: 
  \begin{equation}
  	\mathop{\arg\min}\limits_{\textbf{vec}(\mathring{q_i})} \bigg\{ ||\mathbf{z} - \mathbf{B} \cdotp \textbf{vec}(\mathring{q_i})||_2^2 + \mu ||\textbf{vec}(\mathring{q_i})||_1 \bigg\}. 
  	\label{eqC.13}
  \end{equation}
  Although each $\mathring{q_i}$ contains $(ab)^2$ elements, based on the transformations shown as equation (\ref{eqA.3}) and equation (\ref{eqC.2}), most elements are known to be zero. When solving the lasso problem equation (\ref{eqC.13}), only the non-zero elements of $\textbf{vec}(\mathring{q_i})$ and the corresponding rows and columns of $\mathbf{z}$ and $\mathbf{B}$ are considered. This reduces the number of unknowns in $\textbf{vec}(\mathring{q_i})$ from $(ab)^2$ to $s^2$. Similarly, the optimisation problem equation (\ref{eqQ}) can also be reformulated in the form of equation (\ref{eqC.13}).

Analogously, we converte the optimisation problem in equation (\ref{eqQ}) to a problem of Lasso regression.


\begin{thebibliography}{plainnat}

\bibitem {fan1991optimal} Fan, J.  (1991). On the optimal rates of convergence for nonparametric deconvolution problems. {\em Ann. Stat.}, {\bf 19}, 1257--1272.

\bibitem {levin2009understanding} Levin, A., Weiss, Y., Durand, F. and Freeman, W. T.  (2009). Understanding and evaluating blind deconvolution algorithms.  {\em IEEE Trans. Pattern Analysis and Machine Intelligence}, {\bf 33}, 1964-1971.

\bibitem {rafanelli1993ful} Rafanelli, Gerard L and Rehfield, Mark J   (1993). Full aperture image synthesis using rotating strip aperture image measurements  {\em US Patent 5,243,351 }.

\bibitem {krishnan2009fastl} Krishnan, Dilip and Fergus, Rob  (2009). Fast image deconvolution using hyper-Laplacian priors. {\em Advances in neural information processing systems}, {\bf 22}, 1033--1041.

\bibitem {sun2023characterization} Sun, Yu and Zhi, Xiyang and Zhang, Lei and Jiang, Shikai and Shi, Tianjun and Wang, Nan and Gong, Jinnan (2023).
  Characterization and experimental verification of the rotating synthetic aperture optical imaging system. {\em Scientific Reports}, {\bf 13}, 17015.

\bibitem {osher2017low}Osher, Stanley and Shi, Zuoqiang and Zhu, Wei (2017).
  Low dimensional manifold model for image processing,
{\em SIAM Journal on Imaging Sciences}, {\bf 10},
  1669--1690.


\bibitem{weinberger2006unsupervised} Weinberger, Kilian Q and Saul, Lawrence K (2006).
  Unsupervised learning of image manifolds by semidefinite programming, {\em International journal of computer vision},
  {\bf 70}, 77--90.

\bibitem {yao2023manifold} Yao, Zhigang and Su, Jiaji and Li, Bingjie (2023). Manifold Fitting,
  {\em arXiv preprint arXiv:2304.07680}.  

\bibitem{simoncelli1999bayesian} Simoncelli, Eero P (1999).  {Bayesian denoising of visual images in the wavelet domain}, at {\em Lecture Notes in Statistics}, Springer,  pp.291-308.


\bibitem{zhou2021rotated }Zhou, H., Chen, Y., Feng, H.,  Lv, G., Xu, Z. and Li, Q. (2021). 
  {Rotated rectangular aperture imaging through multi-frame blind deconvolution with Hyper-Laplacian priors},
  {\em Optics Express},
  {\bf 29},
  {12145--12159}.


\bibitem{kingma2014adam}Kingma, Diederik P and Ba, Jimmy (2014).
  {Adam: A method for stochastic optimization},
  {\em arXiv preprint arXiv:1412.6980}.


\bibitem{zuo2016learning}Zuo, W.,  Ren, D., Zhang, D.,  Gu, S. and Zhang, L. (2016). 
  {Learning iteration-wise generalized shrinkage--thresholding operators for blind deconvolution},
  {\em IEEE Transactions on Image Processing},
  {\bf 25},
  {1751--1764}.

\bibitem{satish2020comprehensive}Satish, P.,  Srikantaswamy, M. and Ramaswamy, N. (2020). 
  {A Comprehensive Review of Blind Deconvolution Techniques for Image Deblurring},
{\em Traitement du Signal},
 {\bf 37}, 527-539.

\bibitem{fefferman2016testing}Fefferman, C., Mitter, S. and Narayanan, H. (2016). 
  {Testing the manifold hypothesis},
 {\em Journal of the American Mathematical Society},
 {\bf 29},
  {983--1049},

\bibitem{gong2010locally} Gong, D., Sha, F. and Medioni, G. (2010).
  {Locally linear denoising on image manifolds},
{\em Proceedings of the Thirteenth International Conference on Artificial Intelligence and Statistics},
  {pp. 265--272},

\bibitem{dong2012multi} Dong, W.,  Feng, H.,  Xu, Z. and Li, Qi (2012).
  {Multi-frame blind deconvolution using sparse priors},
{\em Optics Communications},
  {\bf 285},
  {2276--2288},

\bibitem{rameshan2012joint} Rameshan, R., Chaudhuri, S. and Velmurugan, R. (2012).
  {Joint MAP estimation for blind deconvolution: When does it work?},
{\em Proceedings of the Eighth Indian Conference on Computer Vision, Graphics and Image Processing},
 {pp. 1--7}.


\bibitem{miskin2000ensemble} Miskin, J. and MacKay, D. (2000). 
  {Ensemble learning for blind image separation and deconvolution},
 {\em Advances in independent component analysis}, 
 {pp. 123--141},
{Springer.}

\bibitem{cho2009fast}Cho, S. and Lee, S. (2009).
  {Fast motion deblurring},
{\em ACM SIGGRAPH Asia 2009 papers},
{pp. 1--8},

\bibitem{zuo2013generalized} Zuo, W.,  Meng, D.,  Zhang, L., Feng, X. and Zhang, D. (2013). 
  {A generalized iterated shrinkage algorithm for non-convex sparse coding},
{\em Proceedings of the IEEE international conference on computer vision},
 {pp. 217--224}.

\bibitem{zhi2021imaging}Zhi, X.,  Jiang, S.,  Zhang, L.,  Wang, D.,  Hu, J. and Gong, J. (2021).
  {Imaging mechanism and degradation characteristic analysis of novel rotating synthetic aperture system},
{\em Optics and Lasers in Engineering},
 {\bf139},
  {106500}.


\bibitem{stepp2003estimating}Stepp, L.,  Daggert, L. and Gillett, P. (2003). 
  {Estimating the costs of extremely large telescopes},
{\em Future giant telescopes},
  {\bf 4840},
  {309--321}.

\bibitem{zhi2021multi}Zhi, X., Jiang, S.,  Zhang, L.,  Hu, J.,  Yu, L.,  Song, X. and Gong, J. (2021). 
  {Multi-frame image restoration method for novel rotating synthetic aperture imaging system},
{\em Results in Physics},
  {\bf 23},
  {103991}.

\bibitem{sun2024image}Sun, Y.,  Zhi, X.,  Jiang, S.,  Fan, G.,  Yan, X. and Zhang, W. (2024).
  {Image fusion for the novelty rotating synthetic aperture system based on vision transformer},
{\em Information Fusion},
  {\bf 104},
 {102163}.

\bibitem{lv2021full}Lv, Guomian and Xu, Hao and Feng, Huajun and Xu, Zhihai and Zhou, Hao and Li, Qi and Chen, Yueting (2021).
  {A full-aperture image synthesis method for the rotating rectangular aperture system using fourier spectrum restoration},
{\em Photonics},
  {\bf 8},
  {522}.

\bibitem{monreal2018wide}Monreal, Benjamin and Rodriguez, Christian and Carney, Ama and Halliday, Robert and Wang, Mingyuan (2018)
  {Wide Aperture Exoplanet Telescope: a low-cost flat configuration for a 100+ meter ground based telescope},
 {\em Journal of Astronomical Telescopes, Instruments, and Systems},
  {\bf 4},
{024001--024001}.

\bibitem{jelenek2016testing}Jel{\'e}nek, Jan and Kopa{\v{c}}kov{\'a}, Veronika and Kouck{\'a}, Lucie and Mi{\v{s}}urec, Jan (2016).
  {Testing a modified PCA-based sharpening approach for image fusion},
{\em Remote Sensing},
{\bf 8},
{794}.
 

\bibitem{milgrom2020extended}Milgrom, Benjamin and Avrahamy, Roy and David, Tal and Caspi, Avi and Golovachev, Yosef and Engelberg, Shlomo (2020).
  {Extended depth-of-field imaging employing integrated binary phase pupil mask and principal component analysis image fusion},
 {\em Optics Express},
  {\bf 28},
{23862--23873}.


\bibitem{taylor1974principal}Taylor, MM (1974).
  {Principal components colour display of ERTS imagery},
{\em NASA. Goddard Space Flight Center 3d ERTS-1 Symp., Vol. 1, Sect. B},
 PAPER-I12.

\bibitem{kwarteng1989extracting}Kwarteng, P and Chavez, A (1989).
  {Extracting spectral contrast in Landsat Thematic Mapper image data using selective principal component analysis},
 {\em Photogramm. Eng. Remote Sens},
 {\bf 55},
{339--348}.


\bibitem{cheng2015remote}Cheng, Jian and Liu, Haijun and Liu, Ting and Wang, Feng and Li, Hongsheng (2015). 
 {Remote sensing image fusion via wavelet transform and sparse representation},
{\em ISPRS journal of photogrammetry and remote sensing},
 {\bf 104},
  {158--173}.

\bibitem{li1995multisensor} Li, Hui and Manjunath, BS and Mitra, Sanjit K (1995).
  {Multisensor image fusion using the wavelet transform},
{\em Graphical models and image processing},
 {\bf 57},
{235--245}.


\bibitem{amolins2007wavelet} Amolins, Krista and Zhang, Yun and Dare, Peter (2007).
{Wavelet based image fusion techniques—An introduction, review and comparison},
{\em ISPRS Journal of photogrammetry and Remote Sensing},
  {\bf 62},
{249--263}.


\bibitem{zackay2016proper}Zackay, Barak and Ofek, Eran O and Gal-Yam, Avishay (2016). 
 {Proper image subtraction—optimal transient detection, photometry, and hypothesis testing},
{\em The Astrophysical Journal},
{\bf 830},
{27}.

\bibitem{gregson2013stochastic}Gregson, James and Heide, Felix and Hullin, Matthias B and Rouf, Mushfiqur and Heidrich, Wolfgang (2013).
{Stochastic deconvolution},
{\em Proceedings of the IEEE Conference on Computer Vision and Pattern Recognition},
 {pp. 1043--1050}.


\bibitem{babacan2008variational}Babacan, S Derin and Molina, Rafael and Katsaggelos, Aggelos K (2008).
{Variational Bayesian blind deconvolution using a total variation prior},
{\em IEEE Transactions on Image Processing},
{\bf 18},
{12--26}.

\bibitem{wipf2014revisiting}Wipf, David and Zhang, Haichao (2014).
{Revisiting Bayesian blind deconvolution},
{\em Journal of Machine Learning Research},
{\bf 15},
{3775--3814}.

\bibitem{dong2011piecewise} Dong, Wende and Feng, Huajun and Xu, Zhihai and Li, Qi (2011). 
{A piecewise local regularized Richardson--Lucy algorithm for remote sensing image deconvolution},
{\em Optics \& Laser Technology},
 {\bf 43},
{926--933}.

\bibitem{schneider2013benchmarking}Schneider, Dorian and van Ekeris, Tilo and zur Jacobsmuehlen, Joschka and Gross, Sebastian (2013). 
{On benchmarking non-blind deconvolution algorithms: a sample driven comparison of image de-blurring methods for automated visual inspection systems},
{\em IEEE International Instrumentation and Measurement Technology Conference (I2MTC)},
  {pp. 1646--1651}.

\bibitem{kotera2013blind}Kotera, Jan and {\v{S}}roubek, Filip and Milanfar, Peyman (2013).
 {Blind deconvolution using alternating maximum a posteriori estimation with heavy-tailed priors},
{\em Computer Analysis of Images and Patterns: 15th International Conference, CAIP 2013, York, UK, August 27-29, 2013, Proceedings, Part II 15},
  {pp. 59--66}.

\bibitem{sroubek2011robust} Sroubek, Filip and Milanfar, Peyman (2011).
 {Robust multichannel blind deconvolution via fast alternating minimization},
 {\em IEEE Transactions on Image processing},
 {\bf 21},
{1687--1700}.


\bibitem{pan2016l_0}Pan, Jinshan and Hu, Zhe and Su, Zhixun and Yang, Ming-Hsuan (2016).
 {$ l\_0 $-regularized intensity and gradient prior for deblurring text images and beyond},
{\em IEEE transactions on pattern analysis and machine intelligence},
  {\bf 39},
 {342--355}.

\bibitem{krishnan2011blind}Krishnan, Dilip and Tay, Terence and Fergus, Rob (2011).
 {Blind deconvolution using a normalized sparsity measure},
{\em CVPR 2011},
  {233--240}.


\bibitem{gu2014weighted}Gu, Shuhang and Zhang, Lei and Zuo, Wangmeng and Feng, Xiangchu (2014).
{Weighted nuclear norm minimization with application to image denoising},
{\em Proceedings of the IEEE conference on computer vision and pattern recognition},
  {pp. 2862--2869}.


\bibitem{he2010single} He, Kaiming and Sun, Jian and Tang, Xiaoou (2010).
  {Single image haze removal using dark channel prior},
{\em IEEE transactions on pattern analysis and machine intelligence},
{\bf 33},
{2341--2353}.


\bibitem{ren2016image}Ren, Wenqi and Cao, Xiaochun and Pan, Jinshan and Guo, Xiaojie and Zuo, Wangmeng and Yang, Ming-Hsuan (2016).
{Image deblurring via enhanced low-rank prior},
{\em IEEE Transactions on Image Processing},
 {\bf 25},
{3426--3437}.


\bibitem{liu2014blind} Liu, Guangcan and Chang, Shiyu and Ma, Yi (2014).
{Blind image deblurring using spectral properties of convolution operators},
 {\em IEEE Transactions on image processing},
 {\bf 23},
{5047--5056}.
 

\bibitem{chan1998total}Chan, Tony F and Wong, Chiu-Kwong (1998).
  {Total variation blind deconvolution},
{\em IEEE transactions on Image Processing},
{\bf 7},
{370--375}.


\bibitem{delbracio2015removing} Delbracio, Mauricio and Sapiro, Guillermo (2015).
 {Removing camera shake via weighted fourier burst accumulation},
{\em IEEE Transactions on Image Processing},
 {\bf 24},
 {3293--3307}.

\bibitem{ren2020neural}Ren, Dongwei and Zhang, Kai and Wang, Qilong and Hu, Qinghua and Zuo, Wangmeng (2020).
  {Neural blind deconvolution using deep priors},
{\em Proceedings of the IEEE/CVF conference on computer vision and pattern recognition},
  {pp. 3341--3350}.


\bibitem{kotera2021improving}{Kotera, Jan and {\v{S}}roubek, Filip and {\v{S}}m$\acute{\iota}$dl, V{\'a}clav} (2021).
{\em Improving neural blind deconvolution},
 {\em IEEE International Conference on Image Processing (ICIP)},
  {pp. 1954--1958}.


\bibitem{pmlr-v80-massias18a}Massias, Mathurin and Gramfort, Alexandre and Salmon, Joseph (2018).
 {Celer: a Fast Solver for the Lasso with Dual Extrapolation},
  {\em Proceedings of the 35th International Conference on Machine Learning}, {\bf 80},
  {3321--3330}.


\bibitem{massias2020dual} 
  {Mathurin Massias and Samuel Vaiter and Alexandre Gramfort and Joseph Salmon} (2020).
   {\em Dual Extrapolation for Sparse GLMs},
 {\em Journal of Machine Learning Research},
  {\bf 21},
 {1-33}.

\bibitem{wang2013projection}Wang, Weiran and Carreira-Perpin{\'a}n, Miguel A (2013).
  {Projection onto the probability simplex: An efficient algorithm with a simple proof, and an application},
{\em arXiv preprint arXiv:1309.1541}.


\end{thebibliography}

\end{document}